\documentclass[lettersize,journal]{IEEEtran}
\usepackage{amsmath,amssymb,amsfonts}
\usepackage{algorithmic}
\usepackage{algorithm}
\usepackage{array}
\usepackage[caption=false,font=normalsize,labelfont=sf,textfont=sf]{subfig}
\usepackage{xcolor}
\usepackage{textcomp}
\usepackage{stfloats}
\usepackage{url}
\usepackage{verbatim}
\usepackage{graphicx}
\usepackage{cite}
\usepackage{caption}
\usepackage{subcaption}

\def\BibTeX{{\rm B\kern-.05em{\sc i\kern-.025em b}\kern-.08em
		T\kern-.1667em\lower.7ex\hbox{E}\kern-.125emX}}

\newtheorem{lemma}{\it {\textbf{Lemma}}}
\newtheorem{theorem}{\it {\textbf{Theorem}}}
\newtheorem{remark}{\it {\textbf{Remark}}}
\newtheorem{proposition}{\it {\textbf{Proposition}}}
\newenvironment{proof}{{\indent \indent \it {\textbf{Proof}}:}}{\hfill $\blacksquare$ \par}

\allowdisplaybreaks[3]
\begin{document}
	
	\title{Low-Complexity Iterative Precoding Design for Near-field Multiuser Systems With Spatial Non-Stationarity\\}
	
\author{Mengyu Liu, Cunhua Pan, \IEEEmembership{Senior Member, IEEE}, Kangda Zhi, Hong Ren, \IEEEmembership{Member, IEEE}, \\  Cheng-Xiang Wang, \IEEEmembership{Fellow, IEEE}, Jiangzhou Wang, \IEEEmembership{Fellow, IEEE}, and Yonina C. Eldar, \IEEEmembership{Fellow, IEEE}
	\thanks{M. Liu, C. Pan, H. Ren, C. -X. Wang, J. Wang are with National Mobile Communications Research Laboratory, Southeast University, Nanjing 211189, China. (e-mail:{mengyuliu, cpan, hren, chxwang, j.z.wang}@seu.edu.cn). K. Zhi is with the School of Electrical Engineering and Computer Science, Technical University of Berlin, Berlin 10623 (e-mail: k.zhi@tu-berlin.de). Y. C. Eldar is with the Faculty of Mathematics and Computer Science, Weizmann Institute of Science, 7610001 Rehovot, Israel (e-mail: yonina.eldar@weizmann.ac.il).
		
		 Corresponding authors: Cunhua Pan and Hong Ren.}
}
	\maketitle
	
 \begin{abstract}
	Extremely large antenna arrays (ELAA) are regarded as a promising technology for supporting sixth-generation (6G) networks. However, the large number of antennas significantly increases the computational complexity in precoding design, even for linearly regularized zero-forcing (RZF) precoding. To address this issue, a series of low-complexity iterative precoding are investigated. The main idea of these methods is to avoid matrix inversion of RZF precoding. Specifically, RZF precoding is equivalent to a system of linear equations that can be solved by fast iterative algorithms, such as random Kaczmarz (RK) algorithm. Yet, the performance of RK-based precoding algorithm is limited by the energy distributions of multiple users, which restricts its application in ELAA-assisted systems. To accelerate the RK-based precoding, we introduce the greedy random Kaczmarz (GRK)-based precoding by using the greedy criterion-based selection strategy. To further reduce the complexity of the GRK-based precoding, we propose a visibility region (VR)-based orthogonal GRK (VR-OGRK) precoding that leverages near-field spatial non-stationarity, which is characterized by the concept of VR. Next, by utilizing the information from multiple hyperplanes in each iteration, we extend the GRK-based precoding to the aggregation hyperplane Kaczmarz (AHK)-based pecoding algorithm, which further enhances the convergence rate. Building upon the AHK algorithm, we propose a VR-based orthogonal AHK (VR-OAHK) precoding to further reduce the computational complexity. Furthermore, the proposed iterative precoding algorithms are proven to converge to RZF globally at an exponential rate. Simulation results show that the proposed algorithms achieve faster convergence and lower computational complexity than benchmark algorithms, and yield very similar performance to the RZF precoding.
\end{abstract}

\vspace{-0.4cm}
\section{Introduction}

\IEEEPARstart{S}{ixth}-generation (6G) wireless communication networks are expected to offer much higher data rates and nearly ubiquitous coverage compared to fifth-generation (5G) wireless communication networks\cite{you2021towards,8766143,10054381,9237116}. To meet these growing demands, several emerging technologies have been investigated, including millimeter-wave communication \cite{6732923,7400949}, terahertz communication \cite{6005345,8387211}, reconfigurable intelligent surface (RIS) \cite{9475160,9424177,9326394}, and extremely large antenna arrays (ELAA) \cite{10098681,BJORNSON20193}. Among these technologies, ELAA has recently garnered rapidly growing interest \cite{9170653,8869705}. By assembling a large number of antennas, ELAA achieves a high beam gain to compensate for the path loss associated with high signal frequency. Furthermore, with its large physical dimension, ELAA is expected to be integrated into large structures, such as shopping malls, smart factories, and high-speed rail stations, simultaneously supporting high-speed data transmission for a large number of devices \cite{9170651}.

The spectral efficiency of ELAA-assisted communications in hotspot scenarios with dense users could be limited if severe multiuser interference is not well tackled. In conventional massive MIMO systems, linear precoding schemes, such as zero-forcing (ZF) precoding and regularized zero-forcing precoding (RZF), are widely utilized \cite{6415388, 8169014}. Due to the
matrix inversion and the calculation of Gram matrix, the computational complexity
of such linear precoding schemes is on the order of $\mathcal{O}(N_tK^2+K^3)$, where $N_t$ and $K$ denote the number of antennas and users, respectively. Therefore, for ELAA-supported hotspot scenarios with very large number of $N_t$ and $K$, even linear precoding could result in a prohibitive computational burden. To reduce the computational complexity, several low-complexity precoding algorithms have been proposed \cite{7248580,8064675,8425997,8417575,7146060,10220155}.
The main objective of these algorithms is to approximate matrix inversion in RZF precoding, achieving a balance between performance and computational complexity.

 The low-complexity algorithms can be classified into two categories: truncated polynomial expansion (TPE) algorithms and iterative algorithms. In the former, the matrix inversion is approximated by summing a limited series of matrix multiplications. For example, in \cite{7248580} and \cite{8064675}, the Neumann Series expansion was utilized to approximate the matrix inversion in RZF precoding. However, even with this algorithm, the series of matrix multiplications still incur relatively high computational complexity. In contrast, low-complexity iterative precoding algorithms are promising due to their faster convergence rate than TPE. Specifically, the iterative algorithms first transforms the RZF precoding design problem into solving a system of linear equations. Then, a series of iterative methods for solving systems of linear equations is applied to achieve low-complexity precoding designs, including the random Kaczmarz (RK) algorithm \cite{8425997}, the Jacobi iteration \cite{8417575}, the Richardson iteration \cite{7146060}, and the Gauss Seidel iteration \cite{10220155}. The RK algorithm is widely used to solve large systems of linear equations due to its simplicity of application and fast convergence, making it particularly suitable for ELAA-assisted systems. The main idea of the RK algorithm is to project the solution for the current iteration onto the hyperplane defined by the selected equation \cite{10.1007/11830924_45}. Based on \cite{10.1007/11830924_45}, the authors of \cite{8425997} proposed the RK algorithm-based precoding algorithm to effectively reduce the computational complexity where, in each iteration, one of the systems of equations is selected based on a pre-defined probability distribution depending on the energy of the equation.

The performance of the RK-based precoding algorithm in \cite{8425997} is limited in the following two cases: (1) the channel vector energies of the users are similar, and (2) when users are located at the edge of the cell, the low-energy users will be selected with small probability \cite{9437708,strohmer2009comments,censor2009note}. In addition, due to the large number of antennas and users in the ELAA-assisted multiuser system, it is essential to further improve convergence rate of existing RK algorithms and further reduce the computational complexity in the precoding design. Hence, we focus on designing low-complexity and fast-converging RK-based precoding algorithms, making it feasible and promising for ELAA-assisted multiuser systems.

 Besides, compared with existing RK-based precoding algorithms, another difference that should be considered  is the near-field spatial non-stationarity introduced by ELAA. The Rayleigh distance, defined as $ \frac{2D^2}{\lambda}$, indicates the boundary between the near-field and far-field regions, where $D$ and $\lambda$ represent the array aperture and carrier wavelength, respectively. The large array aperture of ELAA will significantly enlarge the Rayleigh distance, extending it from a few meters to several hundred meters. Consequently, for ELAA-assisted communication systems, near-field spherical wave channel models should be considered instead of far-field planar wave channel models \cite{10220205,lu2023tutorial,10225614}. Due to the propagation property of near-field spherical wave, different parts of ELAA may have different views of the propagation environment, which is called the spatial non-stationarity \cite{10500425}. Specifically, the signal power received by the user may be primarily contributed by a portion of ELAA, termed as visibility region (VR), while the contribution of signal strength from other parts of the ELAA could be marginal due to attenuation and obstacles. In \cite{9145378}, an RK-based uplink detection algorithm was investigated with VR. The authors of \cite{10279327} analyzed the interference between subarrays by considering spatial non-stationarity and proposed an RK-based precoding design algorithm. However, existing works mainly focused on the impact of VR on channel modeling, overlooking the potential of leveraging VR information to reduce computational complexity in precoding design.

To overcome the limitations of the existing RK-based precoding and further reduce the computational complexity, we explore a low-complexity near-field precoding design for multiuser ELAA-based systems. The main contributions of this paper are summarized as follows:
\begin{itemize}
	\item  To overcome existing limitations and accelerate convergence of the RK-based precoding, we introduce the greedy random Kaczmarz (GRK)-based precoding, which selects hyperplane with larger residual instead of that with larger energy in each iteration. Next, to further reduce the computational complexity of computing the residuals in each iteration, we propose a VR-based orthogonal greedy random Kaczmarz (VR-OGRK)-based precoding by utilizing the spatial non-stationarity introduced by ELAA. Specifically, non-overlapping VRs ensure channel vector orthogonality between users, which significantly reduces the number of residuals computations per iteration. Furthermore, we show that the VR-OGRK algorithm converges globally at an exponential rate in terms of the normalized mean squared error (NMSE).
	\item Subsequently, by introducing the concept of an aggregation hyperplane, we extend GRK-based precoding to aggregation hyperplane Kaczmarz (AHK)-based precoding. In each iteration, the AHK algorithm utilizes information from multiple hyperplanes, rather than relying on a single hyperplane, thereby effectively accelerates the convergence rate.  Then, we derive that when the AHK algorithm is applied to an orthogonal set, which contains the hyperplanes that are orthogonal to each other, it benefits from both a fast convergence rate and the potential for parallel computation. To further reduce the complexity, we propose a VR-based orthogonal AHK (VR-OAHK) algorithm. Specifically, we partition the users into orthogonal and non-orthogonal sets based on an undirected graph derived from VRs of users. In each iteration, the AHK algorithm is applied sequentially to the orthogonal set and non-orthogonal set. Besides, it is proved that the VR-OAHK algorithm converges globally at an exponential rate.
	\item Simulation results are provided to demonstrate the effectiveness of the proposed low-complexity algorithms. It is shown that the performances of the algorithms are close to that of RZF precoding, while the computational complexity is significantly reduced compared to existing RK-based precoding design.
\end{itemize}

The remainder of this paper is organized as follows. In Section II, we introduce the system model with spatial non-stationarity, and formulate the precoding design problem. In Section III, we briefly introduce an existing RK-based precoding algorithm. In Section IV, by exploiting the features of VR information, we propose the VR-OGRK-based precoding and demonstrate its globally exponential convergence. We then propose the VR-OAHK-based precoding in Section V and prove its globally exponential convergence. Section VI presents the computational complexity analysis for the proposed algorithms. Finally, Sections VII and VIII present simulation results and conclusions, respectively.

\emph{Notations}: In this paper, scalars, vectors, and matrices are represented by lowercase, bold lowercase, and bold uppercase letters, respectively. The identity matrix and the all-zero vector are denoted by $\mathbf{I}$ and $\mathbf{0}$, respectively. A complex matrix of dimensions $M \times N$ is denoted as ${{\mathbb{C}}^{M \times N}}$. The operation of expectation is represented by ${\mathbb{E}}({\cdot})$. The function ${\rm{diag}}(\cdot)$ denotes the diagonalization operation. The absolute value and the real part of a complex number $a$ are denoted by $|a|$ and ${\rm{Re}(a)}$, respectively. The Euclidean norm of a vector $\mathbf{a}$ is represented as ${\| \mathbf{a} \|}_2$. For a matrix $\mathbf{A}$, ${\rm{Tr}}(\mathbf{A})$ refers to its trace, and $\| \mathbf{A} \|_F$ represents its Frobenius norm. The Hadamard product of two matrices $\mathbf{A}$ and $\mathbf{B}$ is expressed as $\mathbf{A} \odot \mathbf{B}$. A vector following a normal distribution with zero mean and unit covariance matrix is denoted as ${\cal C}{\cal N}(\mathbf{0}, \mathbf{I})$. The conjugate, transpose, and Hermitian operations are represented by ${\left( \cdot \right)^{*}}$, ${\left( \cdot \right)^{\rm{T}}}$, and ${\left( \cdot \right)^{\rm{H}}}$, respectively. 
\section{System Model}

\begin{figure}[t]
	\centering
	\includegraphics[width=0.5\linewidth]{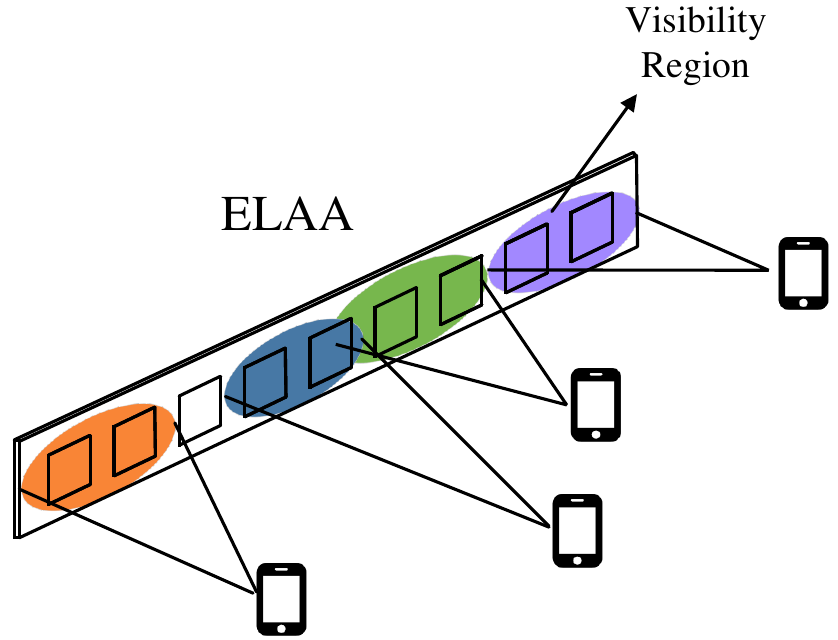}
	\captionsetup{name={Fig.},labelsep=period,singlelinecheck=off}
	\caption{{ELAA-based near-field multiuser system with spatial non-stationarities. }}\vspace{-0.4cm}
	\label{fig1}
\end{figure}

We consider a downlink ELAA-based near-field multiuser system, as illustrated in Fig. \ref{fig1}. In this system, the ELAA with an $N_t$-antenna uniform linear array (ULA) simultaneously serves $K\ (K<N_t)$ single-antenna users. The antenna spacing of the ELAA is denoted by $d=\frac{\lambda}{2}=\frac{c}{2f}$, where $\lambda$, $f$, and $c$ are the carrier wavelength, frequency, and the speed of light, respectively. For tractability, we set the center of the ELAA as the origin with the $n$-th antenna located at $(nd, 0)$, where $n=-\frac{N_t-1}{2}, \cdots, \frac{N_t-1}{2}$.
\vspace{-0.3cm}
\subsection{Near-Field Channel Model with Spatial Non-Stationarity}
Due to the large aperture of the antenna array, we assume that all users and scatterers are distributed within the near-field region of the ELAA. Consequently, the conventional far-field planar wave approximation is no longer valid, and a channel model based on spherical waves should be used to accurately capture the wireless propagation characteristics in the near field.
The near-field channel between user $k$ and the base station (BS) is modeled as \cite{10220205}
\begin{align}
	{\bar{\mathbf{h}}}_{k}=\sqrt{\frac{N_t}{L_k}}({\beta_{k,1}\mathbf{a}(\theta_{k,1},{\mathbf{d}}_{k,1})}+{\sum_{{{l}}=2}^{L_k} \beta_{k,l}\mathbf{a}(\theta_{k,l},{\mathbf{d}}_{k,l})}),
\end{align}
where $\theta_{k,l}$, ${\mathbf{d}}_{k,l}=[ {\mathbf{d}}_{k,l}^{(-\frac{N_t-1}{2})}, \cdots, {\mathbf{d}}_{k,l}^{(0)}, \cdots ,   {\mathbf{d}}_{k,l}^{(\frac{N_t-1}{2})} ]^{\rm{T}} \in \mathbb{C} ^ {N_t \times 1}$, and $L_k$ denotes the angle of arrival (AoA) of the $l$-th path, the vector containing distances between the $l$-th scatterer (user) and all antennas at BS, and the number of paths within the channel of user $k$, respectively. For ease of representation, the first path ($l=1$) represents the line-of-sight (LoS) path. The path of $2\leq l \leq L_k$ is the non-LoS path contributed by the $l$-th scatterer of user $k$, whose location is denoted by $({{\mathbf{d}}^{(0)}_{k,l}}{\rm{cos}}{\theta_{k,l}}, {{\mathbf{d}}^{(0)}_{k,l}}{\rm{sin}}{\theta_{k,l}})$, where $k=1, \cdots, K$. Based on the law of cosines, the distance between the $l$-th scatterer and the $n$-th antenna is calculated as ${{\mathbf{d}}^{(n)}_{k,l}}=\sqrt{{({\mathbf{d}}^{(0)}_{k,l}})^2+n^2d^2-2nd{{\mathbf{d}}^{(0)}_{k,l}}{\rm{cos}}\theta_{k,l}}$.
The near-field array response vector
$\mathbf{a}(\theta_{k,l}, {\mathbf{d}}_{k,l}) \in \mathbb{C}^{N_t \times 1}$ is expressed as
\begin{align}
	\mathbf{a}(\theta_{k,l}, {\mathbf{d}}_{k,l}) =\frac{1}{\sqrt{N_t}}&[e^{-j\frac{2\pi }{{\lambda}}({{\mathbf{d}}_{k,l}}^{(-\frac{N_t-1}{2})}-{{\mathbf{d}}^{(0)}_{k,l}})},\cdots,  \nonumber \\& e^{-j\frac{2\pi }{{\lambda}}{{(\mathbf{d}}_{k,l}}^{(\frac{N_t-1}{2})}-{{\mathbf{d}}^{(0)}_{k,l}})}]^{\rm{T}}.
\end{align}

In addition to the near-field spherical wave effect, the ELAA introduces spatial non-stationarities, making users receive signals only from a small portion of the whole antenna array \cite{9170651}. This phenomenon is characterized by VR, as illustrated in Fig. \ref{fig1}. Specifically, the spatial non-stationarity of the ELAA is mainly attributed to the following two reasons: 1. the large physical size of the entire array causes the channel amplitude to vary across the array and therefore the majority of the signal power is contributed by  proximity subarrays; 2. depending on the location and shape of obstacles between the ELAA and the user, significant portions of the antenna array may be affected by blockage, forming different VRs for different users \cite{yuan2022spatial,9786750}. Considering the spatial non-stationarity effect in the near field, the channel between ELAA and user $k$ is rewritten as
\begin{align}\label{channel}
	\mathbf{h}_{k}=	{\bar{\mathbf{h}}}_{k} \odot \mathbf{u}_k,
\end{align}
where $\mathbf{u}_k \in \mathbb{C}^{N_t \times 1}$  denotes the indicator vector of array VR, which is defined as follows:
\begin{equation} \label{indicate}
	[\mathbf{u}_k]_n=\left\{
	\begin{aligned}
		&1, n\in \Gamma_k,\\
		&0, n \notin \Gamma_k,
	\end{aligned}
	\right.
\end{equation}
where $\Gamma_k$ represents the set of the indices of antennas that are visible to user $k$.

Based on the experimental measurements in \cite{yuan2022spatial}, the whole array of ELAA can be evenly divided into several subarrays, each exhibiting spatial stationarity. Thus, in this paper, we assume that the whole array is divided into $S$ subarrays. Specifically, each subarray possesses $\frac{N_{\rm{t}}}{S}$ antennas, which is assumed to be an integer. Therefore, the VR indicator vector (\ref{indicate}) can be rewritten  with respect to subarrays as
\begin{equation} \label{VR}
	[\mathbf{u}_k]_{\frac{(s-1)N_t}{S}+1: \frac{sN_t}{S}}=\left\{
	\begin{aligned}
		&\mathbf{1}_{\frac{N_t}{S}}, \quad {\rm{if}} \   s \in \mathcal{M}_k,\\
		&\mathbf{0}_{\frac{N_t}{S}},\quad {\rm{else}},
	\end{aligned}
	\right.
\end{equation}
where $\mathcal{M}_k$ denotes the set of the indices of subarrays that are visible to user $k$.
\vspace{-0.6cm}
\subsection{Signal Transmission Model}
The signal received at user $k$ is represented as
\begin{align}
	{{y}}_k&=\sqrt{\rho}\mathbf{h}_k^{\rm{H}} \mathbf{F} \mathbf{s}+n_k=\sqrt{\rho}\mathbf{h}_k^{\rm{H}} \sum_{k=1}^{K}{\mathbf{f}}_{k}{s}_k+n_k,
\end{align}
where $\mathbf{F}=[{\mathbf{f}}_{1},\cdots,{\mathbf{f}}_{k},\cdots,{\mathbf{f}}_{K}] \in \mathbb{C}^{N_t \times K}$ represents the digital precoding matrix.  $\mathbf{s}=[s_{1},\cdots,s_{K}]^{\rm{T}} \in \mathbb{C}^{K\times 1}$ denotes the transmitted symbol satisfying $\mathbb{E}[\mathbf{s}\mathbf{s}^{\rm{H}}]=\mathbf{I}_{K}$. $\rho$ represents the average transmitted signal power and $n_k$ represents thermal noise at the user, which follows a complex Gaussian distribution $\mathcal{C}\mathcal{N}(0,\sigma^2)$ with $\sigma^2$ denoting the noise power.
The spectral efficiency (bit/s/Hz) of user $k$ is then calculated as
\begin{equation} \label{rate}
	{{R}_{k}}	= {\rm{log}}_2\left(1+\frac{\rho|\mathbf{h}_{k}^{\rm{H}}{\mathbf{f}}_{k}|^2}{\rho\sum_{i=1,i \neq k}^{K}|\mathbf{h}_{k}^{\rm{H}}{\mathbf{f}}_{i}|^2+\sigma^2}\right).
\end{equation}
\vspace{-0.8cm}
\subsection{Problem Formulation}
We consider the problem of designing the precoding matrix $\mathbf{F}$ to maximize the sum spectral efficiency (\ref{rate}), which is formulated as
\vspace{-0.2cm}
\begin{subequations} \label{new_optimization problem}
	\begin{align}
		&\mathop {\max }\limits_{\mathbf{F}} \quad
		 \sum_{k=1}^{K}{{R}_{k}}	
		\\
		& \textrm{s.t.}\qquad
		 \|\mathbf{F}\|^2_F \le 1,  \label{power}
	\end{align}
\end{subequations}
where constraint (\ref{power}) denotes the power budget of the BS. 

Linear precoding designs, such as ZF and RZF, have been widely utilized to address (\ref{new_optimization problem}), since they are effective with low complexity in conventional massive MIMO systems. The main idea of ZF is to project the precoding vector of user $k$ onto the null space of the remaining users' channels, ensuring that the transmitted signal to the target user is free from interference from other users. RZF offers a compromise between  interference suppression and noise enhancement.
Specifically, the precoding matrix is computed as
\begin{align}\label{precoding}
	\mathbf{F}^{\textrm{RZF}}=\beta{\mathbf{H}}(\mathbf{H}^{\rm{H}}{\mathbf{H}}+\xi{\mathbf{I}_{K}})^{-1},
\end{align}
where $\beta$, $\xi=\frac{1}{\rm{SNR}}=\frac{\sigma^2}{\rho}$ denote the power normalized factor for meeting power constraint (\ref{power}) and the regularization factor, respectively. $\mathbf{H}=[\mathbf{h}_1,\cdots,\mathbf{h}_K]\in \mathbb{C}^{N_t \times K}$ denotes the channel matrix of all users.
Note that as $\xi \rightarrow \infty$, RZF precoding reduces to maximum ratio combining (MRC) scheme. Conversely, when $\xi =0$, RZF precoding simplifies to ZF precoding.

However, due to matrix inversion, the computational complexity of the RZF precoding is $\mathcal{O}(N_tK^2+K^3)$, which increases with the number of antennas and users. Thus, considering ELAA-assisted hotspot areas with dense users, the computational complexity associated with the RZF precoding could be prohibitive, which limits the practical implementation of ELAA. To address this challenge, a series of iterative-based precoding methods are investigated\cite{8425997,8417575,7146060,10220155}. Among them, the RK-based precoding is widely used due to its simplicity of application and fast convergence. Thus, in this paper, we first briefly introduce the RK-based precoding algorithm and analyze the limitations of the RK-based precoding algorithm in later sections.
\vspace{-0.3cm}
\section{ Low-complexity Random Kaczmarz Algorithm-based Precoding Design}
In this section, we present RK-based precoding algorithm \cite{8425997,10279327} as a benchmark. Subsequently, we analyze the limitations of  existing RK-based precoding algorithm.
\vspace{-0.5cm}
\subsection{Principle of Random Kaczmarz (RK) Algorithm}
The Kaczmarz algorithm is an iterative method suitable for solving large systems of linear equations\cite{10.1007/11830924_45,9145378 }. Specifically, given a system of linear equations $\mathbf{A}\mathbf{x}=\mathbf{b}$, where $\mathbf{A}=[\mathbf{a}^{\rm{H}}_1;\cdots;\mathbf{a}^{\rm{H}}_M]\in \mathbb{C}^{M \times N}$, $\mathbf{x} \in \mathbb{C}^{N \times 1}$, and $\mathbf{b} \in \mathbb{C}^{M \times 1}$. At the $t$-th iteration, the $i$-th equation $\mathbf{a}^{\rm{H}}_i\mathbf{x}^t={b}_i$ is selected and the solution for the next iteration $\mathbf{x}^{(t+1)}$ is updated as \cite{10.1007/11830924_45}
\begin{align}\label{update}
	\mathbf{x}^{(t+1)}=\mathbf{x}^{(t)}+\frac{{b}_i-\mathbf{a}^{\rm{H}}_i\mathbf{x}^{(t)}}{\|\mathbf{a}^{\rm{H}}_i\|^2_2}\mathbf{a}_i.
\end{align}
From (\ref{update}), the main idea of the RK algorithm is to project $\mathbf{x}^{t}$ onto the hyperplane defined by the selected equation $\mathbf{a}^{\rm{H}}_i\mathbf{x}^t=b_i$. At each iteration, the RK algorithm selects the hyperplane based on the pre-determined probability distribution. The pre-determined probability distribution, often referred to as the selection strategy, is typically defined as $\mathbf{p}_i=\frac{\|\mathbf{a}^{\rm{H}}_i\|^2_2}{\|\mathbf{A}\|_F^2}$ (energy criterion) or $\mathbf{p}_i=\frac{1}{M}$ (uniform criterion), $\forall i \in [1,M]$, where $\mathbf{p}_i$ denotes the probability of selecting the $i$-th row of the systems of linear euqations. 
\vspace{-0.3cm}
\subsection{Random Kaczmarz (RK)-based precoding design }
By appling the RZF precoding (\ref{precoding}), the received signal of all users is represented as follows:
\begin{align}\label{symbols}
	\mathbf{y}={\mathbf{H}}^{\rm{H}}\mathbf{F}^{\rm{RZF}}\mathbf{s}+ \mathbf{n}  =&\beta{\mathbf{H}}^{\rm{H}}{\mathbf{H}}(\mathbf{H}^{\rm{H}}{\mathbf{H}}+\xi{\mathbf{I}_{K}})^{-1}\mathbf{s} +\mathbf{n} \nonumber \\ =&\beta{\mathbf{H}}^{\rm{H}}{\mathbf{H}}\mathbf{v} +\mathbf{n},
\end{align}
where $\mathbf{v}=(\mathbf{H}^{\rm{H}}{\mathbf{H}}+\xi{\mathbf{I}_{K}})^{-1}\mathbf{s} \in \mathbb{C}^{K \times 1}$ denotes the auxiliary vector. Subsequently, the acquisition of auxiliary vector $\mathbf{v}=(\mathbf{H}^{\rm{H}}{\mathbf{H}}+\xi{\mathbf{I}_{K}})^{-1}\mathbf{s}$ is readily transformed into an equivalent least square problem, as described by the following theorem.
\begin{theorem} \label{theorem1}
	The auxiliary vector $\mathbf{v}=(\mathbf{H}^{\rm{H}}{\mathbf{H}}+\xi{\mathbf{I}_{K}})^{-1}\mathbf{s}$  can be obtained from an equivalent least square problem
	\begin{align} \label{ls}
		\mathbf{v}= \arg \min\limits_{\mathbf{v}} \| \mathbf{G}\mathbf{v} -\mathbf{g}\|^2_2,
	\end{align}
	where $\mathbf{G}=[\mathbf{H};\sqrt{\xi} {\mathbf{I}}_K] \in \mathbb{C}^{(N_t+K) \times K}$, $\mathbf{g}=[\mathbf{0}_{N_t};\frac{\mathbf{s}}{\sqrt{\xi}}] \in \mathbb{C}^{{(N_t+K) \times 1}}$.
\end{theorem}

\begin{proof}
	See Appendix A.
\end{proof}
According to (\ref{symbols}) and (\ref{ls}), if we get $\mathbf{v}$, the precoded signal can be easily obtained by $\mathbf{H}\mathbf{v}=\mathbf{F}\mathbf{s}$. Thus, RZF precoding (\ref{precoding}) is equivalent to finding $\mathbf{v}$ in the least square problem of (\ref{ls}).
 If $\mathbf{g}$ lies in the range space spanned by $\mathbf{G}$, Problem (\ref{ls}) can be further reformulated as a system of linear equations, i.e., $\mathbf{G}\mathbf{v}=\mathbf{g}$, which can be effectively solved by the RK algorithm as we discussed in Section III. A.  However, this system of linear equations is inconsistent unless $\mathbf{s}=\mathbf{0}$. This inconsistency makes it challenging to obtain an accurate solution for $\mathbf{v}$. In other words, the system of linear equations $\mathbf{G}\mathbf{v}=\mathbf{g}$ may not have a solution.

 To tackle this issue, the inconsistency in the systems of linear equations needs to be removed \cite{Randomized}. Specifically, we first decompose $\mathbf{g}$ as follows
\begin{align}
	\mathbf{g}=\mathbf{w}+\boldsymbol{\psi},
\end{align}
with $\mathbf{w} \in \mathcal{R}(\mathbf{G}), \boldsymbol{\psi} \in \mathcal{R}(\mathbf{G})^{\perp}=\mathcal{K}(\mathbf{G}^{\rm{H}})$, where $\mathcal{R}(\cdot)$, $\mathcal{R}(\cdot)^{\perp}$ and $\mathcal{K}(\cdot)$ represents the range space, the orthogonal complementary space of the range space, and the kernel space. Thus, if we can accurately estimate $\mathbf{w}$ from $\mathbf{g}$, the solution of the system of linear equations $\mathbf{G}\mathbf{v}={\mathbf{w}}$ can be approximated as the solution of Problem (\ref{ls}).
To accurately estimate $\mathbf{w}$, we establish a system of linear equations as follows
\begin{align}\label{SLE}
	\mathbf{G}^{\rm{H}}\mathbf{w}=	\mathbf{G}^{\rm{H}}(\mathbf{g}-\boldsymbol{\psi})=\mathbf{G}^{\rm{H}}\mathbf{g}=\mathbf{s}.
\end{align}
\begin{algorithm}[t]
	\caption{{ RK-based precoding algorithm }}\label{iteda1}
	\begin{algorithmic}[1]
		\STATE {\bf{Input:}} $\mathbf{H}, N_t, K, \xi$, the tolerance error $\epsilon$.
		\STATE Initialize $\mathbf{V}=\mathbf{0}_{K\times K}$, $t=0$;
		\STATE Initialize $\mathbf{m}_k=0$, $\mathbf{q}_k=0, \forall k \in [1,K]$;
		\STATE Calculate and store $\|\mathbf{h}_i\|^2_2+\xi$;
		\WHILE{\rm{NMSE} $\ge \epsilon$ }
		\STATE Run the RK algorithm $K$ times: 
		\STATE Pick the $i$-th row $\mathbf{H}^{\rm{H}}$ based on the pre-determined probability distribution;
		\STATE Calculate $[\mathbf{r}^{(t)}_k]_i={[\mathbf{e}_k]_i-\mathbf{h}^{\rm{H}}_i\mathbf{m}_k^{(t)}-\xi [\mathbf{q}_k]_i^{(t)}}$;
		\STATE Calculate $[\boldsymbol{\gamma}^{(t)}_k]_i=[\mathbf{r}^{(t)}_k]_i/({\|\mathbf{h}^{\rm{H}}_i\|^2_2+\xi})$;
		\STATE Update $\mathbf{m}_k^{(t+1)}=\mathbf{m}_k^{(t)}+[\boldsymbol{\gamma}^{(t)}_k]_i\mathbf{h}_i$;
		\STATE Update $[\mathbf{q}_k]_i^{(t+1)}=[\mathbf{q}_k]_i^{(t)}+[\boldsymbol{\gamma}^{(t)}_k]_i$;
		\STATE $t=t+1$;
		\ENDWHILE
		\STATE  $\mathbf{V}=[\mathbf{q}_1,\cdots,\mathbf{q}_K]$, Calculate ${\mathbf{F}}$ as (\ref{precoding1});
		\STATE {\bf{Output:}} ${\mathbf{F}}$, $T=t$.
	\end{algorithmic}
\end{algorithm}
By using the RK algorithm to solve the system of linear equations $\mathbf{G}^{\rm{H}}\mathbf{w}=\mathbf{s}$, we obtain the estimated value ${\mathbf{w}}$. Note that, ${\mathbf{w}}$ lies in the space spanned by the columns of $\mathbf{G}$, as indicated by (\ref{update}). Next, the consistent over-determined system of linear equations is formulated as
\begin{align}\label{SLE2}
	\mathbf{G}\mathbf{v}={\mathbf{w}}.
\end{align}
Finally, we obtain the value of $\mathbf{v}$ by solving this equation. Based on the above analysis, the solution of the least square problem (\ref{ls}) can be obtained in two steps. Specifically, in Step 1, we solving the system of linear equation (\ref{SLE}) by using the RK algorithm to get the estimated value $\mathbf{w}$. In Step 2, after obtaining $\mathbf{w}$, the system of linear euqations (\ref{SLE2}) is also solved by using RK algorithm. Furthermore, we can conclude the following remarks.

\begin{remark}
	We first define the estimated value ${\mathbf{w} \in \mathbb{C}^{(N_t+K)\times 1}}$ as $[\mathbf{m};\sqrt{\xi}\mathbf{q}]$, where vector $\mathbf{q}$ is of length $K$. Due to the special structure of $\mathbf{G}$, we have
	\begin{align}
		\mathbf{G}\mathbf{v}=[\mathbf{H};\sqrt{\xi} {\mathbf{I}}_K]\mathbf{v}=[\mathbf{H}\mathbf{v};\sqrt{\xi}\mathbf{v}].
	\end{align}
Thus, the estimated solution of $\mathbf{v}$ is readily determined by extracting the scaled last $K$ elements of $\mathbf{w}$, i.e., $\mathbf{q}$, rather than solving the system of linear equations (\ref{SLE2}).
\end{remark}

\begin{remark}
	In the downlink precoding design, the objective is to compute $\mathbf{F}^{\rm{RZF}} $. However, due to the relationship $\mathbf{H}\mathbf{v}=\mathbf{F}^{\rm{RZF}} \mathbf{s}$, we can only obtain the product $\mathbf{F}^{\rm{RZF}} \mathbf{s}$, given solution $\mathbf{v}$. Therefore, to derive the precoding $\mathbf{F}^{\rm{RZF}} $, we propose to use the parallel RK algorithm as introduced in \cite{8425997,9145378}. Specifically, we run the RK algorithm in parallel $K$ times by setting $\mathbf{s}_k=\mathbf{e}_k \in \mathbb{C}^{K \times 1}$, where $\mathbf{e}_k$ denotes the $k$-th canonical basis. Then, the precoding matrix can then be computed as
	\begin{align}\label{precoding1}
		\mathbf{F}^{\rm{RZF}} =\mathbf{H}\times[\mathbf{v}_1,\cdots,\mathbf{v}_K]=\mathbf{H}\mathbf{V}.
	\end{align}
\end{remark}
Based on the above discussions, the RK-based precoding algorithm is detailed in Algorithm \ref{iteda1}, which reduces the computational complexity of RZF from $\mathcal{O}(N_tK^2+K^3)$ to $\mathcal{O}(TN_t)$. In Algorithm \ref{iteda1}, we apply the RK algorithm to (\ref{SLE}), i.e., $\mathbf{G}^{\rm{H}}\mathbf{w}=\mathbf{s}$. Then, $\mathbf{v}$ is updated based on Remark 1. Specifically, Steps 8 and 9 are to calculate the update weights, which is defined as $\frac{{b}_i-\mathbf{a}^{\rm{H}}_i\mathbf{x}^{(t)}}{\|\mathbf{a}_i\|^2_2}$ in (\ref{update}). In Steps 10 and 11, the value of solution $\mathbf{w}=[\mathbf{m};\sqrt{\xi}\mathbf{q}]$ is updated according to (\ref{update}).
\vspace{-0.8cm}
\subsection{Limitations of the RK-based Precoding Algorithm}                         
 As shown in Algorithm \ref{iteda1}, the pre-determined probability distribution of the RK algorithm determines which hyperplane is selected at each iteration. Thus, the convergence rate is mainly depending on the design of the pre-determined probability distribution. The existing RK-based precoding algorithms \cite{8425997,10279327} usually utilize the energy criterion-based selection strategy, which has a faster convergence rate than uniform criterion-based selection strategy \cite{10.1007/11830924_45}.
  Specifically, as discussed in Section III. B, the RK algorithm is applied to solve the system of linear equations $\mathbf{G}^{\rm{H}}\mathbf{w}=\mathbf{s}$. Recall that $\mathbf{G}=[\mathbf{H};\sqrt{\xi} {\mathbf{I}}_K]$, the physical significance of the RK-based precoding algorithm is that the hyperplane defined in terms of the user channel vector is selected at each iteration based on the selection strategy. Thus, existing RK-based precoding algorithms select the high-energy user with a higher probability in each iteration.
  
  However, the performance of existing RK-based precoding algorithms utilizing energy criterion-based selection strategy  is limited in the following two cases: (1) when users are located at the edge of the cell, and (2) the channel vector energies of the users are similar. For the first case, the RK-based precoding algorithm always selects users with higher energy and users with lower-energy are difficult to be selected. For the second case, the energy criterion will degenerate into a uniform criterion. In both of the above scenarios, the convergence speed of the RK-based precoding algorithm is drastically reduced. Both of the above cases are common in the ELAA-assisted hotspot scenarios in future 6G networks. 
  Furthermore, existing RK-based precoding algorithms do not consider near-field spatial non-stationarity introduced by ELAA. According to (\ref{channel}) and (\ref{VR}), the channel vectors of users have a large number of zero elements. By fully exploiting the near-field non-stationrity, the complexity of RK-based precoding can be further reduced. Thus, to overcome these limitations and further reduce the computational complexity, we propose low-complexity iterative precoding algorithms.  
\vspace{-0.3cm}
\section{VR-based Orthogonal Greedy Random Kaczmarz algorithm}
In this section, we first introduce the greedy random Kaczmarz (GRK)-based precoding algorithm by designing the greed criterion-based selection strategy. Subsequently, to further reduce the computational complexity and exploit the spatial non-stationarity, we propose the VR-based Orthogonal GRK (VR-OGRK) algorithm and provide the convergence analysis of the proposed algorithm.
\vspace{-0.4cm}
\subsection{Greedy Random Kaczmarz (GRK)-based precoding}
As discussed in Section III. C, the energy criterion-based selection strategy is not applicable to the RK-based precoding algorithm. Therefore, we introduce the greedy random Kaczmarz (GRK)-based precoding algorithm by using the greedy criterion-based selection strategy. Specifically,
Different from existing RK-based precoding, the GRK-based precoding algorithm employs a selection strategy that prioritizes the $i$-th row with larger residual, i.e., with larger $r_i={b_i-\mathbf{a}^{\rm{H}}_i\mathbf{x}^{(t)}}$ \cite{GRIEBEL20121596,doi:10.1137/17M1137747}. The greedy selection strategy of GRK-based precoding is defined as
\begin{align}
\mathbf{p}_i=\frac{|{r}_i|^2}{\|\mathbf{r}\|_2^2},
\end{align} 
where $\mathbf{r}=[r_1;\cdots,r_i,\cdots;r_M]$ denotes the residual vector. 

The GRK-based precoding effectively overcomes the limitations of existing RK-based precoding and accelerate the convergence. However, the GRK-based precoding algorithm introduces additional complexity in calculating all residuals $\mathbf{r}$ at each iteration. To further reduce the complexity, we propose the VR-based orthogonal GRK (VR-OGRK)-based precoding algorithm by using the spatial non-stationarity.
\vspace{-0.4cm}
\subsection{VR-based Orthogonal GRK (VR-OGRK)-based precoding}
To effectively exploit the spatial non-stationary brought by the ELAA-assisted system, we introduce the following lemma.
\begin{lemma}
	If the VRs of user $i$ and user $j$ do not overlap, the channels between the two users are mutually orthogonal, i.e. $\mathbf{h}^{\rm{H}}_i\mathbf{h}_j=0$, if $\mathcal{M}_i \cap \mathcal{M}_j = \emptyset$.
\end{lemma}
\begin{proof}
	This can be readily proved by calculating the inner product of the two users' channel vectors.
\end{proof}

\begin{algorithm}[t]
	\caption{{VR-OGRK-based precoding algorithm}}\label{iteda3}
	\begin{algorithmic}[1]
		\STATE {\bf{Input:}} $\mathbf{H}$, $\mathbf{H}_s$, $N_t$, $K$, $\xi$, $\mathcal{X}_k$, $\mathcal{M}_k$, $\mathcal{Q}_s$, error threshold $\epsilon$.
		\STATE Repeat Algorithm \ref{iteda1} Steps $2$-$4$;
		\STATE Initialize the residuals $\mathbf{r}_k=\mathbf{e}_k,\  \forall  k \in [1,K]$
		\STATE Calculate the probability vector $[\mathbf{p}_k]_i=|[\mathbf{r}_k]_i|^2/\|\mathbf{r}_K\|_2^2, \   \forall i, k \in [1,K]$;
		\STATE Construct $\widetilde{\mathbf{h}}_k=[\mathbf{h}_k]_{\Gamma_k}$;
		\WHILE{\rm{NMSE} $\ge \epsilon$ }
		\STATE Run VR-OGRK algorithm $K$ times: 
		\STATE$[\mathbf{r}^{(t)}_k]_{{\mathcal{X}}_{i^{(t-1)}}}=[\mathbf{e}_k]_{{\mathcal{X}}_{i^{(t-1)}}}-\widetilde{\mathbf{h}}^{\rm{H}}_{{\mathcal{X}}_{i^{(t-1)}}}\mathbf{m}_k^{(t)}-\xi [\mathbf{q}^{(t)}_k]_{{\mathcal{X}}_{i^{(t-1)}}}$;
		\STATE Update the probability vector as Step 4;
		\STATE Pick the $i^{(t)}$-th row $\mathbf{H}^{\rm{H}}$ according to $\mathbf{p}_k$;
		\STATE Repeat Algorithm \ref{iteda1} Steps $8$-$12$;
		\ENDWHILE
		\STATE  $\mathbf{V}=[\mathbf{q}_1,\cdots,\mathbf{q}_K]$, Calculate ${\mathbf{F}}$ as (\ref{precoding2});
		\STATE {\bf{Output:}} ${\mathbf{F}}$, $T=t$.
	\end{algorithmic}
\end{algorithm}

Based on Lemma 1, we establish a non-orthogonal set $\mathcal{X}_k$ corresponding to user $k$, where the elements of this set represent the indices of users whose VRs overlap with that of user $k$. For the system of linear equations $\mathbf{G}^{\rm{H}}\mathbf{w}=\mathbf{s}$, the $k$-th row of $\mathbf{G}^{\rm{H}}$ can be expressed as $[\mathbf{h}_k,\sqrt{\xi}\mathbf{e}_k]$. Thus, if VRs of user $i$ and user $j$ do not overlap, the $i$-th row and the $j$-th row of $\mathbf{G}^{\rm{H}}$ are also mutually orthogonal. Theorem 2 elaborates how to apply the orthogonality in Kaczmarz algorithm.

\begin{theorem}
	Consider the linear equations $\mathbf{A}\mathbf{x}=\mathbf{b}$, solved by the Kaczmarz algorithm. At the $t$-th iteration, if the $i$-th row of $\mathbf{A}$ is selected, the residuals corresponding to the rows that are orthogonal to the $i$-th row will remain unchanged, i.e.,
	\begin{align}
		{{r}}_j^{(t+1)}=	{{r}}^{(t)}_j, {\rm{if} }\ \mathbf{a}^{\rm{H}}_j \mathbf{a}_i=0,
	\end{align}
	where $	{{r}}_j^{(t)}={{b}_j-\mathbf{a}^{\rm{H}}_j\mathbf{x}^{(t)}}$ denotes the residual corresponding to the $j$-th row at the $i$-th iteration.
\end{theorem}
\begin{proof}
	See Appendix B.
\end{proof}

 According to Lemma 1 and Theorem 2, when selecting the $k$-th row of $\mathbf{G}^{{\rm{H}}}$, only the residuals associated with the rows in the non-orthogonal set $\mathcal{X}_k$ are affected.

 After building the set $\mathcal{X}_k$ based on VR information, we can only update the residuals corresponding to the rows whose indices are in $\mathcal{X}_k$, rather than all rows in each iteration. This method effectively reduces the complexity of the GRK-based precoding algorithm. Furthermore, the operation $\mathbf{a}_i^{\rm{H}}\mathbf{{x}}$ incurs a computational cost of $\mathcal{O}(N_t)$, which brings the primary computation burden at each iteration. Due to the presence of spatial non-stationarity, we leverage the VR information of users to further reduce the computational cost in this operation. Specifically, we construct the effective channel $\widetilde{\mathbf{h}}_k=[\mathbf{h}_k]_{\Gamma_k}\in \mathbb{C}^{|\Gamma_k|\times 1}$. Thus, the complexity of the original operation is further reduced to $\mathcal{O}(|\Gamma|)$, assuming that $|\Gamma_k| = |\Gamma|$ for all users, which simplifies the analysis.

In Step 14 of Algorithm \ref{iteda1}, the computational cost of $\mathbf{F}=\mathbf{H}\mathbf{V}$ scales on the order of $\mathcal{O}(N_tK^2)$. However, existing literature usually ignored the complexity introduced by this term, which leads to significant computational burden in precoding design. Therefore, we will again leverage the VR information to further reduce the complexity in this step.

We first define $\mathbf{H}_s \in \mathbb{C}^{ \frac{N_t}{S}\times K}$ and  $\mathcal{Q}_s$ as the channel matrix between the $s$-th subarray and users and the set of indices of users whose VRs include the $s$-th subarray, respectively. Then, based on the VR information $\mathcal{Q}_s$, we construct the equivalent matrix $[\mathbf{V}]_{\mathcal{Q}_s} \in \mathbb{C}^{|\mathcal{Q}_s| \times K}$ and  $[\mathbf{H}_s]_{\mathcal{Q}_s} \in \mathbb{C}^{\frac{N_t}{S} \times |\mathcal{Q}_s|}$. The precoding matrix is further calculated as
\begin{align} \label{precoding2}
	{\mathbf{F}}=[[\mathbf{H}_1]_{\mathcal{Q}_1} [\mathbf{V}]_{\mathcal{Q}_1}; \cdots; [\mathbf{H}_S]_{\mathcal{Q}_S} [\mathbf{V}]_{\mathcal{Q}_S}].
\end{align}
 For analytical tractability, we assume that $ |\mathcal{Q}_s| = |\mathcal{Q}| $ for all subarrays. Then, it is observed that the complexity of calculating $\mathbf{F}$ based on (\ref{precoding2}) is reduced from $O(N_tK^2)$ to $\mathcal{O}(N_t{|\mathcal{Q}|}K)$, which is beneficial for ELAA systems.

Based on the above discussions, the overall VR-OGRK-based precoding is detailed in Algorithm \ref{iteda3}. In each iteration, a row is randomly selected based on the selection strategy determined by the relative value of residual, as described in Step 4. The residuals corresponding to the rows whose indices are in $\mathcal{X}_k$ are updated, as shown in Step 8.
\vspace{-0.5cm}
\subsection{Convergence Analysis}
For the subsequent convergence analysis, we first present the following lemma:
\begin{lemma}
	In the VR-OGRK algorithm, given a system of linear equations $\mathbf{A}\mathbf{x}=\mathbf{b}$, where $\mathbf{A}\in \mathbb{C}^{M \times N}$, $\mathbf{x} \in \mathbb{C}^{N \times 1}$, the vector $(\mathbf{x}^{(t+1)}-\mathbf{x}^{(t)})$ is perpendicular to the vector $(\mathbf{x}^{(t+1)}-\mathbf{x}^{*})$, i.e.
 	\begin{align}\label{lemma1}
	 (\mathbf{x}^{(t+1)}-\mathbf{x}^{*})^{\rm{H}}(\mathbf{x}^{(t+1)}-\mathbf{x}^{(t)})=0,
	\end{align}
	where $\mathbf{x}^{*}$ denotes the optimal solution of the system.
\end{lemma}
	\begin{proof}
		According to (\ref{update}), we have
		\begin{align}\label{lemma2}
			(\mathbf{x}^{(t+1)}-\mathbf{x}^{(t)})=\frac{{b}_i-\mathbf{a}^{\rm{H}}_i\mathbf{x}^{(t)}}{\|\mathbf{a}^{\rm{H}}_i\|^2_2}\mathbf{a}_i.
		\end{align}
			From (\ref{lemma2}), the vector $(\mathbf{x}^{(t+1)}-\mathbf{x}^{(t)})$ lies in the subspace spanned by $\mathbf{a}_i$. Thus, to demonstrate the relationship in (\ref{lemma1}), we need to prove $\mathbf{a}^{\rm{H}}_i(\mathbf{x}^{(t+1)}-\mathbf{x}^{*})=0$, as demonstrated below:
		\begin{align}
			&\mathbf{a}^{\rm{H}}_i(\mathbf{x}^{(t+1)}-\mathbf{x}^{*})=\mathbf{a}^{\rm{H}}_i(\mathbf{x}^{(t)}-\mathbf{x}^{*}+\frac{b_i-\mathbf{a}^{\rm{H}}_i\mathbf{x}^{(t)}}{\|\mathbf{a}^{\rm{H}}_i\|^2_2}\mathbf{a}_i) \nonumber \\
			&=\mathbf{a}^{\rm{H}}_i\mathbf{x}^{(t)}-\mathbf{a}^{\rm{H}}_i\mathbf{x}^{*}+{{b}_i-\mathbf{a}^{\rm{H}}_i\mathbf{x}^{(t)}}=0,
		\end{align}
		which completes the proof of Lemma 2.
	\end{proof}

	Based on Lemma 2, we obtain the following relationship by using the Pythagorean theorem
	\begin{align}\label{gougu}
		\|\mathbf{x}^{(t+1)}-\mathbf{x}^{*}\|^2_2 =\|\mathbf{x}^{(t)}-\mathbf{x}^{*}\|^2_2-\|\mathbf{x}^{(t+1)}-\mathbf{x}^{t}\|^2_2.
	\end{align}
Based on (\ref{gougu}), we derive the following theorem to prove the convergence of our proposed algorithm.
\begin{theorem} The VR-OGRK algorithm converges as
	\begin{align}
		\|\mathbf{x}^{(t+1)}-\mathbf{x}^{*}\|^2_2 \le \eta^{t}\|\mathbf{x}^{(0)}-\mathbf{x}^{*}\|^2_2,
	\end{align}
	with global convergence rate of 
	\begin{align}
		\eta=1-\frac{\kappa^2(\mathbf{A})}{\|\mathbf{A}\|_F} < 1,
	\end{align}
	where $\kappa(\mathbf{A})=\min \limits_{\mathbf{x}-\mathbf{x}^* \in \mathcal{R}(\mathbf{A}^{\rm{H}}) }\frac{\|\mathbf{A}(\mathbf{x}-\mathbf{x}^*)\|_{2}}{\|\mathbf{x}-\mathbf{x}^*\|_2}$.
\end{theorem}
\begin{proof}
	First, according to (\ref{gougu}), we have 
	\begin{align}
		&\mathbb{E}[\|\mathbf{x}^{(t+1)}-\mathbf{x}^{*}\|^2_2] =\|\mathbf{x}^{(t)}-\mathbf{x}^{*}\|^2_2-\mathbb{E}[\|\mathbf{x}^{(t+1)}-\mathbf{x}^{(t)}\|^2_2] \nonumber \\
		&=\|\mathbf{x}^{(t)}-\mathbf{x}^{*}\|^2_2-\sum_{i=1}^{M}\frac{|\mathbf{r}_i|^2}{\|\mathbf{r}^{(t)}\|_2^2}\frac{ |{b}_i-\mathbf{a}^{\rm{H}}_i\mathbf{x}^{(t)}|^2}{\|\mathbf{a}_i\|^2_2} \nonumber \\
		& =\|\mathbf{x}^{(t)}-\mathbf{x}^{*}\|^2_2-\frac{1}{\|\mathbf{r}^{(t)}\|_2^2}\sum_{i=1}^{M}\frac{ |\mathbf{r}_i|^4}{\|\mathbf{a}_i\|^2_2}\nonumber \\
		& \le \|\mathbf{x}^{(t)}-\mathbf{x}^{*}\|^2_2-\frac{1}{\|\mathbf{r}^{(t)}\|_2^2}\frac{\sum_{i=1}^{M} |\mathbf{r}_i|^4}{\sum_{i=1}^{M}\|\mathbf{a}_i\|^2_2} \nonumber \\
    	&= \|\mathbf{x}^{(t)}-\mathbf{x}^{*}\|^2_2-\frac{\|\mathbf{r}^{(t)}\|_2^2}{\sum_{i=1}^{M}\|\mathbf{a}_i\|^2_2} \nonumber \\
		& = \|\mathbf{x}^{(t)}-\mathbf{x}^{*}\|^2_2-\frac{\|\mathbf{A}(\mathbf{x}^{(t)}-\mathbf{x}^*)\|_2^2}{\|\mathbf{A}\|^2_F}  \nonumber \\
		& \le (1-\frac{\kappa^2(\mathbf{A})}{\|\mathbf{A}\|^2_F})\|\mathbf{x}^{(t)}-\mathbf{x}^{*}\|^2_2,
	\end{align}
	where $\kappa(\mathbf{A})=\min \limits_{\mathbf{x}-\mathbf{x}^* \in \mathcal{R}(\mathbf{A}^{\rm{H}}) }\frac{\|\mathbf{A}(\mathbf{x}-\mathbf{x}^*)\|_{2}}{\|\mathbf{x}-\mathbf{x}^*\|_2}$ denotes the minimum singular value.
	Define $\eta=\frac{\kappa^2(\mathbf{A})}{\|\mathbf{A}\|^2_F} \le 1$, which represents the global convergence rate.
	
	Hence, the proof of Theorem 3 is completed.
\end{proof}
\vspace{-0.2cm}
\section{VR-based orthogonal aggregation hyperplane Kaczmarz algorithm}
We extend the GRK-based precoding to the aggregation hyperplane Kaczmarz (AHK)-based precoding by exploiting the information from multiple hyperplanes in each iteration. To further reduce computational complexity, we exploit VR information to  propose the VR-based orthogonal aggregation hyperplane Kaczmarz (VR-OAHK)-based precoding and analyze its convergence.
\vspace{-0.5cm}
\subsection{Aggregation Hyperplane Kaczmarz-based Precoding}
As discussed in Section IV, at the $t$-th iteration of the VR-OGRK algorithm, if the $i$-th row of matrix $\mathbf{A}$ is selected, the solution of the $(t+1)$-th iteration is updated as follows
 \begin{align}\label{update3}
 	\mathbf{x}^{(t+1)}&=\mathbf{x}^{(t)}+\frac{{b}_i-\mathbf{a}^{\rm{H}}_i\mathbf{x}^{(t)}}{\|\mathbf{a}^{\rm{H}}_i\|^2_2}\mathbf{a}_i=\mathbf{x}^{(t)}+\frac{{r}^{(t)}_i}{\|\mathbf{a}^{\rm{H}}_i\|^2_2}\mathbf{a}_i.
 \end{align}
Thus, at each iteration, the solution of the next iteration $\mathbf{x}^{(t+1)}$ is updated by projecting the current solution $\mathbf{x}^{(t)}$ into the selected hyperplane, which is defined by $\mathbf{a}^{\rm{H}}_i\mathbf{x}={b}_i$. This means that we only use the information from one hyperplane at each iteration.

 A natural question arises: \textit{is it possible to utilize information from multiple hyperplanes instead of a single hyperplane at each iteration?} If the information from multiple hyperplanes are jointly utilized, the convergence rate can be further enhanced. 
  Inspired by the surrogate constraint proposed in \cite{glover1968surrogate,glover2003tutorial}, which is introduced as a heuristic method to solve integer programming problems and graph optimization problems, we next propose an aggregation hyperplane Kaczmarz (AHK) algorithm to jointly use multiple hyperplane. Specifically, AHK algorithm constructs an aggregation hyperplane by linearly combining multiple hyperplanes with predefined weights. The next iteration solution $\mathbf{x}^{(t+1)}$ is then updated by projecting $\mathbf{x}^{(t)}$ onto the aggregation hyperplane, which incorporates information from multiple hyperplanes.
 
The update criteria of the AHK methods is expressed as
 \begin{align}\label{update4}
 	\mathbf{x}^{(t+1)}
 	&=\mathbf{x}^{(t)}+\frac{(\boldsymbol{\varphi}^{(t)})^{\rm{H}}(\mathbf{b}-\mathbf{A}\mathbf{x}^{(t)})}{\|\mathbf{A}^{\rm{H}}\boldsymbol{\varphi}^{(t)}\|^2_2}\mathbf{A}^{\rm{H}}\boldsymbol{\varphi}^{(t)} \nonumber \\
 	&=\mathbf{x}^{(t)}+\frac{(\boldsymbol{\varphi}^{(t)})^{\rm{H}}\mathbf{r}^{(t)}}{\|\mathbf{A}^{\rm{H}}\boldsymbol{\varphi}^{(t)}\|^2_2}\mathbf{A}^{\rm{H}}\boldsymbol{\varphi}^{(t)},
 \end{align}
 where $\boldsymbol{\varphi}$ denotes the aggregation coefficient vector. Note that, when $\boldsymbol{\varphi}=\mathbf{e}_i$, where $i$ represents the index of the row which is selected at the $t$-th iteration, (\ref{update4}) reduces to the update of traditional Kaczmarz algorithm.
 
In (\ref{update4}), the solution of the next iteration $\mathbf{x}^{(t+1)}$ is updated by projecting the current solution $\mathbf{x}^{(t)}$ onto the aggregation hyperplane which is defined by 
\begin{align}
	(\boldsymbol{\varphi}^{(t)})^{\rm{H}}\mathbf{A}\mathbf{x}=(\boldsymbol{\varphi}^{(t)})^{\rm{H}}\mathbf{b}.
\end{align}
The aggregation hyperplane retains more information of the matrix $\mathbf{A}$ and vector $\mathbf{b}$ compared to GRK algorithm. The AHK algorithm linearly combines multiple hyperplanes into an aggregation hyperplane with predefined weights. Thus, at each iteration, more information from the matrix $\mathbf{A}$ and vector $\mathbf{b}$ is utilized to accelerate convergence. 

Recall that, in GRK algorithm, the residuals of the current solution on all hyperplanes, i.e., $\mathbf{r}^{(t)}$, are obtained at each iteration, as shown in Step 8 of Algorithm \ref{iteda3}. Thus, the aggregation coefficient vector $\boldsymbol{\varphi}^{(t)}$ is defined as
 \begin{align} \label{coefficient}
 [{\boldsymbol{\varphi}}^{(t)}]_i=\frac{{b}_i-\mathbf{a}^{\rm{H}}_i\mathbf{x}^{(t)}}{\|\mathbf{a}_i^{\rm{H}}\|^2_2}, \forall i \in [1,M].
 \end{align}
 \vspace{-0.7cm}
 \subsection{VR-based Orthogonal AHK (VR-OAHK)-based Precoding}
The AHK-based precoding requires the generation of an aggregation hyperplane, introducing additional complexity of $\mathcal{O}(N_tK)$ at each iteration. While this algorithm can substantially accelerate convergence, it may not result in a significant reduction in overall complexity.
 To address this issue, we propose VR-based orthogonal aggregation hyperplane Kaczmarz (VR-OAHK) algorithm. Specifically, by fully exploiting the VR features and further accelerating the convergence, we first divide $K$ users into two sets: the orthogonal set $\mathcal{F}$ and the non-orthogonal set $\mathcal{V}$. The orthogonal set $\mathcal{F}$ consists of users whose VRs do not overlap, i.e. ${|\mathcal{M}_i \cap \mathcal{M}_j |} > 0, \forall i,j \in \mathcal{F}$, whereas the opposite holds true for the non-orthogonal set $\mathcal{V}$. In the following, we prove some promising properties of the proposed AHK algorithm on the orthogonal set, such as faster convergence and parallelization features. In each iteration, we first apply AHK algorithm using information from the hyperplanes corresponding to $\mathcal{F}$. Subsequently, we apply AHK algorithm using information from the hyperplanes corresponding to $\mathcal{V}$.
 
To utilize the potential of exploiting spatial non-stationarity in simplifying algorithms, we derive the following theorem.
\begin{theorem}\label{lemma3}
		Consider using the AHK algorithm to solve a system of linear equations $\mathbf{A}\mathbf{x}=\mathbf{b}$, where $\mathbf{A}\in \mathbb{C}^{M \times N}$, $\mathbf{x} \in \mathbb{C}^{N \times 1}$. We assume that there exists a set of mutually orthogonal hyperplanes in the linear system, denoted as $\mathcal{F}$. At the $t$-th iteration, if the aggregation hyperplane is generated by the hyperplanes in $\mathcal{F}$, the updated solution $\mathbf{x}^{(t+1)}$ will lie in the solution space spanned by the hyperpalnes in $\mathcal{F}$, i.e.,
		\begin{align}
			\mathbf{A}_{[\mathcal{F},:]}\mathbf{x}^{(t+1)}=\mathbf{b}_{\mathcal{F}}.
		\end{align}
\end{theorem}
\begin{proof}
	See Appendix C.
\end{proof}
\begin{algorithm}[t]
	\caption{{VR-OAHK-based precoding algorithm}}\label{iteda4}
	\begin{algorithmic}[1]
		\STATE {\bf{Input:}} $\mathbf{H}$, $\mathbf{H}_s$, $N_t$, $K$, $\xi$, $\mathcal{F}$, $\mathcal{V}$, $\Gamma_k$, $\mathcal{Q}_s$, threshold $\epsilon$.
		\STATE Repeat Algorithm \ref{iteda1} Steps $2$-$4$;
		\STATE Initialize ${\mathbf{r}^{\rm{or}}_k}=\mathbf{0}_K, {\mathbf{r}^{\rm{nor}}_k}=\mathbf{0}_K,{\boldsymbol{\varphi}^{\rm{or}}_k}=\mathbf{0}_K,{\boldsymbol{\varphi}^{\rm{nor}}_k}=\mathbf{0}_K,\  \forall  k \in [1,K]$;
		\STATE Construct $\widetilde{\mathbf{h}}_k=[\mathbf{h}_k]_{\Gamma_k}$;
		\WHILE{\rm{NMSE} $\ge \epsilon$ }
		\STATE Run the VR-OAHK algorithm $K$ times: 
		\STATE For orthogonal user group $\mathcal{F}$:
		\STATE$[{\mathbf{r}^{\rm{or}}_k}^{(t)}]_{\mathcal{F}}=[\mathbf{e}_k]_{\mathcal{F}}-\widetilde{\mathbf{h}}^{\rm{H}}_{\mathcal{F}}\mathbf{m}_k^{(t)}-\xi [\mathbf{q}^{(t)}_k]_{\mathcal{F}}$;
		\STATE$[{\boldsymbol{\varphi}^{\rm{or}}_k}^{(t)}]_{i}=\frac{[{\mathbf{r}^{\rm{or}}_k}^{(t)}]_i}{{\|\mathbf{h}_i\|^2_2+\xi}}, \forall i \in \mathcal{F}$;
		\STATE  $\mathbf{m}_k^{(t+\frac{1}{2})}=\mathbf{m}_k^{(t)}+\sum_{i\in \mathcal{F}}^{|\mathcal{F}|}[{\boldsymbol{\varphi}^{\rm{or}}_k}^{(t)}]_i\mathbf{h}_i$;
		\STATE  $\mathbf{q}_k^{(t+\frac{1}{2})}=\mathbf{q}_k^{(t)}+{\boldsymbol{\varphi}^{\rm{or}}_k}^{(t)}$;
		\STATE For non-orthogonal user group $\mathcal{V}$:
		\STATE$[{\mathbf{r}^{\rm{nor}}_k}^{(t)}]_{\mathcal{V}}=[\mathbf{e}_k]_{\mathcal{V}}-\widetilde{\mathbf{h}}^{\rm{H}}_{\mathcal{V}}\mathbf{m}_k^{(t+\frac{1}{2})}-\xi [\mathbf{q}^{(t+\frac{1}{2})}_k]_{\mathcal{V}}$;
		\STATE$[{\boldsymbol{\varphi}^{\rm{nor}}_k}^{(t)}]_{i}=\frac{[{\mathbf{r}^{\rm{nor}}_k}^{(t)}]_i}{{\|\mathbf{h}_i\|^2_2+\xi}}, \forall i \in \mathcal{V}$;
		\STATE  ${{\gamma}^{\rm{nor}}_k}^{(t)}=\frac{({\boldsymbol{\varphi}^{\rm{nor}}_k}^{(t)})^{\rm{H}}{\mathbf{r}^{\rm{nor}}_k}^{(t)}}{\|\mathbf{H}{\boldsymbol{\varphi}^{\rm{nor}}_k}^{(t)}\|^2_2+\xi\|{\boldsymbol{\varphi}^{\rm{nor}}_k}^{(t)}\|^2_2}$;
		\STATE  $\mathbf{m}_k^{(t+1)}=\mathbf{m}_k^{(t+\frac{1}{2})}+{{\gamma}^{\rm{nor}}_k}^{(t)}\mathbf{H}{\boldsymbol{\varphi}^{\rm{nor}}_k}^{(t)}$;
		\STATE  $\mathbf{q}_k^{(t+1)}=\mathbf{q}_k^{(t+\frac{1}{2})}+{{\gamma}^{\rm{nor}}_k}^{(t)}{\boldsymbol{\varphi}^{\rm{nor}}_k}^{(t)}$;
		\STATE $t=t+1$;
		\ENDWHILE
		\STATE  $\mathbf{V}=[\mathbf{q}_1,\cdots,\mathbf{q}_K]$, Calculate ${\mathbf{F}}$ as (\ref{precoding2});
		\STATE {\bf{Output:}} ${\mathbf{F}}$, $T=t$.
	\end{algorithmic}
\end{algorithm}

According to Theorem 4, if the orthogonal set $\mathcal{F}$ can be established, only one iteration is required to converge to the solution space spanned by the hyperplanes in $\mathcal{F}$, which further accelerates convergence.
Furthermore, the orthogonality among different users in the set $\mathcal{F}$ enables parallelization of the algorithm, as detailed in the following proposition.
\begin{proposition}
	A system of linear equations $\mathbf{A}\mathbf{x}=\mathbf{b}$, where $\mathbf{A}\in \mathbb{C}^{M \times N}$, $\mathbf{x} \in \mathbb{C}^{N \times 1}$, is considered to be solved by the AHK  algorithm. We assume that there exists a set of mutually orthogonal hyperplanes in the linear system, denoted as $\mathcal{F}$. If the aggregation hyperplane is generated by the hyperplanes in $\mathcal{F}$, the solution of the next iteration is updated as
	\begin{align}\label{orthogonal}
		\mathbf{x}^{(t+1)}=\mathbf{x}^{(t)}+\sum_{i \in \mathcal{F}}^{|{\mathcal{F}}|} [{\boldsymbol{\varphi}}^{(t)}]_i\mathbf{a}_i.
	\end{align}
\end{proposition}
\begin{proof}
	See Appendix D.
\end{proof}

 According to (\ref{orthogonal}), the vector $  [{\boldsymbol{\varphi}}^{(t)}]_i\mathbf{a}_i, \forall i \in \mathcal{F}$, can be independently calculated for each hyperplane and then summed to update the next iteration solution $\mathbf{x}^{(t+1)}$. When updating the solution based on the orthogonal set $\mathcal{F}$, the computation of the relevant parameters can be divided into multiple subtasks. Each subtask is then executed in parallel on the multicore CPU at the BS and therefore greatly reduces the running time.
 
Next, we focus on how to accurately construct the orthogonal set $\mathcal{F}$. Specifically, we propose a graph theory-based orthogonal set generation algorithm.
We first construct an undirected graph $G=(V,E)$, where $V=\{v_1,\cdots,v_K\}$ and $E=\{(v_i,v_j)|\ v_i,v_j\in V\}$ denote the set of vertex corresponding with $K$ users and the set of edges, respectively. If $(v_i,v_j) \in E$, then there exists an edge between the vertices $v_i$ and $v_j$. 
By utilizing the VR information of different users, then we establish the edge between two users based on their VR overlap situations. Specifically, if ${|\mathcal{M}_i \cap \mathcal{M}_j |} > 0$, then we establish an edge $(v_i,v_j) \in E$ in the graph. 

After constructing the VR-based undirected graph, the orthogonal set $\mathcal{F}$ is regarded as the independent set in the undirected graph, which ensures that there are no edges between any two vertices in this set. To enlarge the impact of the orthogonal set on improving algorithm performance, our goal is to find as many as possible vertices in the VR-based graph $G=(V,E)$ to form a maximum independent set. Thus, we formulate following maximum independent set problem
\begin{subequations} \label{max_set}
	\begin{align}
		\mathop {\max } \quad
		& |\mathcal{F}|	\\
		\textrm{s.t.}\qquad
		& \mathcal{F} \subseteq V, \\ 
		& (v_i,v_j) \notin E, \forall v_i,v_j \in \mathcal{F}.
	\end{align}
\end{subequations}

Problem (\ref{max_set}) is a classic problem in graph theory, which can be effectively addressed with a low-complexity algorithm with linear complexity $\mathcal{O}(K)$ \cite{10.1145/3035918.3035939}. After constructing the orthogonal set $\mathcal{F}$, the remaining users can naturally be grouped into $\mathcal{V}$. Thus, in the VR-OAHK algorithm, the update criteria at each iteration is
\begin{align}
	&\mathbf{x}^{(t+\frac{1}{2})}=\mathbf{x}^{(t)}+\frac{({\boldsymbol{\varphi}^{\rm{or}}}^{(t)})^{\rm{H}}{\mathbf{r}^{\rm{or}}}^{(t)}}{\|\mathbf{A}^{\rm{H}}{\boldsymbol{\varphi}^{\rm{or}}}^{(t)}\|^2_2}\mathbf{A}^{\rm{H}}{\boldsymbol{\varphi}^{\rm{or}}}^{(t)}, \label{updateor}\\
	&\label{updatenor}	\mathbf{x}^{(t+1)}=\mathbf{x}^{(t+\frac{1}{2})}+\frac{({\boldsymbol{\varphi}^{\rm{nor}}}^{(t)})^{\rm{H}}{\mathbf{r}^{\rm{nor}}}^{(t)}}{\|\mathbf{A}^{\rm{H}}{\boldsymbol{\varphi}^{\rm{nor}}}^{(t)}\|^2_2}\mathbf{A}^{\rm{H}}{\boldsymbol{\varphi}^{\rm{nor}}}^{(t)},
\end{align}
where the aggregation coefficient and residual of the orthogonal set and non-orthogonal set are respectively defined as
\begin{align}
	[{\mathbf{r}^{\rm{or}}}^{(t)}]_{\mathcal{F}}&=\mathbf{b}_{\mathcal{F}}-\mathbf{A}_{[\mathcal{F},:]}\mathbf{x}^{(t)}, \\ \label{or}
	[{\boldsymbol{\varphi}^{\rm{or}}}^{(t)}]_{\mathcal{F}}&=\frac{[{\mathbf{r}^{\rm{or}}}^{(t)}]_i}{\|\mathbf{a}_i^{\rm{H}}\|^2_2}, \forall i \in \mathcal{F}, \\
	[{\mathbf{r}^{\rm{nor}}}^{(t)}]_{\mathcal{V}}&=\mathbf{b}_{\mathcal{V}}-\mathbf{A}_{[\mathcal{V},:]}\mathbf{x}^{(t+\frac{1}{2})}, \\ \label{nor}
	[{\boldsymbol{\varphi}^{\rm{nor}}}^{(t)}]_{\mathcal{V}}&=\frac{[{\mathbf{r}^{\rm{nor}}}^{(t)}]_i}{\|\mathbf{a}_i^{\rm{H}}\|^2_2}, \forall i \in \mathcal{V} .
\end{align}
Based on the above discussions, the detailed procedure is outlined in Algorithm \ref{iteda4}.
\vspace{-0.3cm}
\subsection{Convergence Analysis}
\begin{lemma}
	In the VR-OAHK algorithm, given a system of linear equations $\mathbf{A}\mathbf{x}=\mathbf{b}$, where $\mathbf{A}\in \mathbb{C}^{M \times N}$, $\mathbf{x} \in \mathbb{C}^{N \times 1}$, in each iteration, we have the following relationship:
	\begin{align}\label{lemma5}
		&(\mathbf{x}^{(t+\frac{1}{2})}-\mathbf{x}^{*})^{\rm{H}}(\mathbf{x}^{(t+1)}-\mathbf{x}^{(t+\frac{1}{2})})=0 , \\
		&(\mathbf{x}^{(t)}-\mathbf{x}^{*})^{\rm{H}}(\mathbf{x}^{(t+\frac{1}{2})}-\mathbf{x}^{(t)})=0 ,
	\end{align}
	where $\mathbf{x}^{*}$ denotes the optimal solution of the system.
\end{lemma}
\begin{proof} 
	 According to (\ref{updateor}) and (\ref{updatenor}), we have
	 \vspace{-0.3cm}
\begin{align}
	\mathbf{x}^{(t+1)}-\mathbf{x}^{(t+\frac{1}{2})}&=\frac{({\boldsymbol{\varphi}^{\rm{nor}}}^{(t)})^{\rm{H}}{\mathbf{r}^{\rm{nor}}}^{(t)}}{\|\mathbf{A}^{\rm{H}}{\boldsymbol{\varphi}^{\rm{nor}}}^{(t)}\|^2_2}\mathbf{A}^{\rm{H}}{\boldsymbol{\varphi}^{\rm{nor}}}^{(t)}, \label{lemm31}  \\
	\mathbf{x}^{(t+\frac{1}{2})}-\mathbf{x}^{(t)}&=\frac{({\boldsymbol{\varphi}^{\rm{or}}}^{(t)})^{\rm{H}}{\mathbf{r}^{\rm{or}}}^{(t)}}{\|\mathbf{A}^{\rm{H}}{\boldsymbol{\varphi}^{\rm{or}}}^{(t)}\|^2_2}\mathbf{A}^{\rm{H}}{\boldsymbol{\varphi}^{\rm{or}}}^{(t)}.\label{lemm32}
\end{align}
Thus, to demonstrate (\ref{lemma5}), we need to show that $\mathbf{A}^{\rm{H}}{\boldsymbol{\varphi}^{\rm{nor}}}^{(t)}(\mathbf{x}^{(t+\frac{1}{2})}-\mathbf{x}^{(*)})=0$, which is proved as
\begin{align} \label{prove}
	&	(\mathbf{A}^{\rm{H}}{\boldsymbol{\varphi}^{\rm{nor}}}^{(t)})^{\rm{H}}(\mathbf{x}^{(t+\frac{1}{2})}-\mathbf{x}^{(*)})\nonumber \\& = ({\boldsymbol{\varphi}^{\rm{nor}}}^{(t)})^{\rm{H}}\mathbf{A}(\mathbf{x}^{(t)}-\mathbf{x}^*+\frac{({\boldsymbol{\varphi}^{\rm{nor}}}^{(t)})^{\rm{H}}\mathbf{r}^{(t)}}{\|\mathbf{A}^{\rm{H}}{\boldsymbol{\varphi}^{\rm{nor}}}^{(t)}\|^2_2}\mathbf{A}^{\rm{H}}{\boldsymbol{\varphi}^{\rm{nor}}}^{(t)}) \nonumber \\
	&=({\boldsymbol{\varphi}^{\rm{nor}}}^{(t)})^{\rm{H}}\mathbf{A}\mathbf{x}^{(t)}-({\boldsymbol{\varphi}^{\rm{nor}}}^{(t)})^{\rm{H}}\mathbf{A}\mathbf{x}^*+{({\boldsymbol{\varphi}^{\rm{nor}}}^{(t)})^{\rm{H}}\mathbf{r}^{(t)}} \nonumber 
	\\&=({\boldsymbol{\varphi}^{\rm{nor}}}^{(t)})^{\rm{H}}(\mathbf{A}\mathbf{x}^{(t)}-\mathbf{b})+{({\boldsymbol{\varphi}^{\rm{nor}}}^{(t)})^{\rm{H}}\mathbf{r}^{(t)}}=0.
\end{align}
Following similar steps with (\ref{prove}) , (\ref{lemm32}) can be proved, completing the proof.
\end{proof}

Based on Lemma 3, we obtain the following relationship by using the Pythagorean theorem
\begin{align}\label{gougu2}
	\|\mathbf{x}^{(t+1)}-\mathbf{x}^{*}\|^2_2 &=\|\mathbf{x}^{(t+\frac{1}{2})}-\mathbf{x}^{*}\|^2_2-\|\mathbf{x}^{(t+1)}-\mathbf{x}^{(t+\frac{1}{2})}\|^2_2,  \\ \|\mathbf{x}^{(t+\frac{1}{2})}-\mathbf{x}^{*}\|^2_2 &=\|\mathbf{x}^{(t)}-\mathbf{x}^{*}\|^2_2-\|\mathbf{x}^{(t+\frac{1}{2})}-\mathbf{x}^{(t)}\|^2_2.
\end{align}
	\begin{figure*}[hb] 
	\centering 
	\hrulefill 
	\vspace*{1pt} 
	\newcounter{TempEqCnt} 
	\setcounter{TempEqCnt}{\value{equation}} 
	\setcounter{equation}{50} %
	\begin{align}\label{line}
		&\|\mathbf{e}^{(t+1)}\|^2_2=\|\mathbf{e}^{(t+\frac{1}{2})}\|^2_2-\frac{(\mathbf{e}^{(t+\frac{1}{2})})^{\rm{H}}\mathbf{A}^{\rm{H}}\boldsymbol{\Phi}\mathbf{A}\mathbf{e}^{(t+\frac{1}{2})}\mathbf{A}\mathbf{e}^{(t+\frac{1}{2})}(\mathbf{e}^{(t+\frac{1}{2})})^{\rm{H}}\mathbf{A}^{\rm{H}}\boldsymbol{\Phi}\mathbf{A}\mathbf{e}^{(t+\frac{1}{2})}\mathbf{A}\mathbf{e}^{(t+\frac{1}{2})}}{(\mathbf{e}^{(t+\frac{1}{2})})^{\rm{H}}\mathbf{A}^{\rm{H}}\boldsymbol{\Phi}^{\rm{T}}\mathbf{A}\mathbf{A}^{\rm{H}}\boldsymbol{\Phi}^{\rm{T}}\mathbf{A}\mathbf{e}^{(t+\frac{1}{2})}}  \nonumber \\ &=(1-\frac{((\mathbf{e}^{(t+\frac{1}{2})})^{\rm{H}}\mathbf{A}^{\rm{H}}\boldsymbol{\Phi}\mathbf{A}\mathbf{e}^{(t+\frac{1}{2})}\mathbf{A}\mathbf{e}^{(t+\frac{1}{2})})^2}{(\mathbf{e}^{(t+\frac{1}{2})})^{\rm{H}}(\mathbf{A}^{\rm{H}}\boldsymbol{\Phi}^{\rm{T}}\mathbf{A})^2\mathbf{e}^{(t+\frac{1}{2})}(\mathbf{e}^{(t+\frac{1}{2})})^{\rm{H}}\mathbf{e}^{(t+\frac{1}{2})}})\|\mathbf{e}^{(t+\frac{1}{2})}\|^2_2 \nonumber \\
		&=(1-\frac{\boldsymbol{\vartheta}\mathbf{P}\boldsymbol{\vartheta}^{\rm{H}}}{\boldsymbol{\vartheta}\boldsymbol{\vartheta}^{\rm{H}}})\|\mathbf{e}^{(t+\frac{1}{2})}\|^2_2,
	\end{align}
\end{figure*}
\begin{table*}[ht] 
	\centering
	\caption{Computational Complexity Comparison}
	\begin{tabular}{|c|c|}
		\hline
		\textbf{Scheme} & \textbf{Computational Complexity {[}FLOPS{]}} \\ \hline
		RZF             &  $8K^3+9K^2+12N_tK^2-3K$                                             \\ \hline
		URK             & $8N_tK^2+4N_tK+T^{\rm{URK}}(16N_t-4)$                  \\ \hline
		SWOR-ERK        &  $8N_tK^2+4N_tK+K-1+T^{\rm{ERK}}(16N_t+K+8)$                                              \\ \hline
		GK        &  $8N_tK^2+4N_tK-K+T^{\rm{GK}}(8N_t(K+1)+K-5)$                                              \\ \hline
		VR-OGRK          &  $8N_tK(|{\mathcal{Q}}|+\frac{1}{2})-K+T^{\rm{OGRK}}(8|{\mathcal{M}}|(|{\mathcal{X}}|+1)+K-5)$                                             \\ \hline
		VR-OAHK          &  $8N_tK(|{\mathcal{Q}}|+\frac{1}{2})+T^{\rm{OAHK}}(8|{\mathcal{M}}|(2|{\mathcal{V}}|+|{\mathcal{F}}|+2)+24|\mathcal{V}|+14N_t+5)$                                             \\ \hline
	\end{tabular}\label{table1}
\end{table*}
\begin{theorem}
	The VR-OAHK algorithm converges as
	\newcounter{TempEqCnt1} 
	\setcounter{TempEqCnt1}{\value{equation}} 
	\setcounter{equation}{46} %
		\begin{align}
		\|\mathbf{x}^{(t+1)}-\mathbf{x}^{*}\|^2_2 \le \alpha^{2t}\|\mathbf{x}^{(0)}-\mathbf{x}^{*}\|^2_2,
	\end{align}
	with global convergence rate 
	\begin{align}
		\alpha=(1-\frac{\boldsymbol{\vartheta}\mathbf{P}\boldsymbol{\vartheta}^{\rm{H}}}{\boldsymbol{\vartheta}\boldsymbol{\vartheta}^{\rm{H}}}) < 1,
	\end{align}
	where $\boldsymbol{\vartheta}=\mathbf{A}^{\rm{H}}\boldsymbol{\Phi}^{\rm{T}}\mathbf{A}\mathbf{e}^{(t+\frac{1}{2})}$, $\mathbf{P}=\frac{\mathbf{e}^{(t+\frac{1}{2})}(\mathbf{e}^{(t+\frac{1}{2})})^{\rm{H}}}{(\mathbf{e}^{(t+\frac{1}{2})})^{\rm{H}}\mathbf{e}^{(t+\frac{1}{2})}}$, $\mathbf{e}^{(t+\frac{1}{2})}=\mathbf{x}^{(t+\frac{1}{2})}-\mathbf{x}^{*}$.
\end{theorem}
\begin{proof}
		To simplify the analysis, we define $\mathbf{x}^{(t)}-\mathbf{x}^{*}=\mathbf{e}^{(t)}$. Then, according to (\ref{gougu2}), we have the following relationship
		\begin{align}\label{new1}
		&\|\mathbf{e}^{(t+1)}\|^2_2 =\|\mathbf{e}^{(t+\frac{1}{2})}\|^2_2-\|\mathbf{x}^{(t+1)}-\mathbf{x}^{(t+\frac{1}{2})}\|^2_2\nonumber \\
		&=\|\mathbf{e}^{(t+\frac{1}{2})}\|^2_2-\frac{(\mathbf{e}^{(t+\frac{1}{2})})^{\rm{H}}\mathbf{A}^{\rm{H}}{\boldsymbol{\varphi}^{\rm{nor}}}^{(t)}({\boldsymbol{\varphi}^{\rm{nor}}}^{(t)})^{\rm{H}}\mathbf{A}\mathbf{e}^{(t+\frac{1}{2})}}{\|\mathbf{A}^{\rm{H}}{\boldsymbol{\varphi}^{\rm{nor}}}^{(t)}\|^2_2}.
		\end{align}
	Define $\boldsymbol{\Phi}={\rm{diag}}(\frac{1}{\|\mathbf{a}^{\rm{H}}_1\|_2^2},\cdots,\frac{1}{\|\mathbf{a}^{\rm{H}}_{M}\|_2^2}) \in \mathbb{R}^{M \times M}$. Then, (\ref{nor}) can be rewritten as
	\begin{align}\label{new2}
		{\boldsymbol{\varphi}^{\rm{nor}}}^{(t)}=-\boldsymbol{\Phi}\mathbf{A}\mathbf{e}^{(t+\frac{1}{2})}.
	\end{align}
	
	By substituting (\ref{new2}) into (\ref{new1}), $ \|\mathbf{e}^{(t+1)}\|^2_2$ is reformulated as (\ref{line}), shown at the bottom of this page,
where the parameters are defined as
\vspace{-0.3cm}
	\newcounter{TempEqCnt3} 
\setcounter{TempEqCnt3}{\value{equation}} 
\setcounter{equation}{51} %
	\begin{align}
		\boldsymbol{\vartheta}&=\mathbf{A}^{\rm{H}}\boldsymbol{\Phi}^{\rm{T}}\mathbf{A}\mathbf{e}^{(t+\frac{1}{2})}, \\
		\mathbf{P}&=\frac{\mathbf{e}^{(t+\frac{1}{2})}(\mathbf{e}^{(t+\frac{1}{2})})^{\rm{H}}}{(\mathbf{e}^{(t+\frac{1}{2})})^{\rm{H}}\mathbf{e}^{(t+\frac{1}{2})}},
	\end{align}
	where $\mathbf{P}$ is a rank-one Hermitian matrix, the eigenvalue of $\mathbf{P}$ is either 0 or 1. Then, according to Rayleigh quotient theorem, we can obtain the following relationship
	\begin{align}
		0 < \frac{\boldsymbol{\vartheta}\mathbf{P}\boldsymbol{\vartheta}^{\rm{H}}}{\boldsymbol{\vartheta}\boldsymbol{\vartheta}^{\rm{H}}} \le 1,
	\end{align}
and the convergence rate	$\alpha=(1-\frac{\boldsymbol{\vartheta}\mathbf{P}\boldsymbol{\vartheta}^{\rm{H}}}{\boldsymbol{\vartheta}\boldsymbol{\vartheta}^{\rm{H}}}) \in [0,1)$. The convergence analysis of orthogonal set is similar to the non-stationary set. Therefore, we have
	\begin{align}
		\|\mathbf{x}^{(t+1)}-\mathbf{x}^{*}\|^2_2 &\le \alpha \|\mathbf{x}^{(t+\frac{1}{2})}-\mathbf{x}^{*}\|^2_2 \nonumber \\ &\le \alpha^2\|\mathbf{x}^{(t)}-\mathbf{x}^{*}\|^2_2  \le
		\alpha^{2t}\|\mathbf{x}^{(0)}-\mathbf{x}^{*}\|^2_2,
	\end{align}
		completing the proof.
\end{proof}
\section{Complexity Analysis}
The detailed computational complexity, expressed in floating point operations (FLOPS) \cite{golub2013matrix}, is presented in Table \ref{table1}. In this table, $T^{\rm{URK}}$, $T^{\rm{ERK}}$, $T^{\rm{GK}}$, $T^{\rm{OGRK}}$, $T^{\rm{OAHK}}$ represent the iteration numbers for the uniform criterion-based RK algorithm (URK) \cite{8425997}, the energy criterion-based RK algorithm with sampling without replacement technology (SWOR-ERK)\cite{10279327}, the GK algorithm \cite{orthogonality}, the VR-OGRK algorithm, and the VR-OAHK algorithm, respectively.  For the sake of analysis, we assume here, $|\Gamma_k| = |\Gamma|$, $|\mathcal{X}_k| = |\mathcal{X}|, \forall k$, and but the true value of $|\Gamma_k|$, $|\mathcal{X}_k|$ will be used in the simulation section. As shown in Table \ref{table1}, for the proposed VR-OGRK algorithm in Algorithm \ref{iteda3}, the computational complexity of calculating the precoding matrix $\mathbf{F}$ reduces from $\mathcal{O}(8N_tK^2)$ to $\mathcal{O}(8N_tK|\mathcal{Q}|)$ and the computational complexity of Step 8 reduces from $\mathcal{O}(8N_tK)$ to $\mathcal{O}(8|\Gamma||\mathcal{X}|)$ at each iteration. Therefore, our algorithm can work effectively in ELAA-supported hotspot scenarios. Furthermore, the main advantage of the two proposed algorithms is the significant reduction in the number of iterations compared to the traditional Kaczmarz-type algorithms, which will be shown in the simulation section.
\section{Simulation Results}
In this section, numerical results are presented to demonstrate the effectiveness of the proposed VR-OGRK algorithm and the VR-OAHK algorithm compared to the traditional Kaczmarz-type algorithm.
\vspace{-0.3cm}
\subsection{Simulation Configuration}
 Without loss of generality, the simulation parameters are assumed to be $N_t=2000$, $f= 100\  \rm{GHz}$, $S=20$, $\epsilon=10^{-6}$, $L_k=5$. There are $K=30$ users randomly distributed within a range of $1$ m to $50$ m from the ELAA. ${\rm{SNR}}_{k}=\frac{\rho\|{{\mathbf{h}}}_{k}\|^2_2}{\sigma^2}$ is defined as the received SNR for user $k$. For simplicity, the SNRs of all users are assumed to be equal due to power control, i.e., ${\rm{SNR}}=0 \ \rm{dB}$. From \cite{9733790} and \cite{9777939}, the VR set $\mathcal{M}_k$ can be modeled as a Bernoulli distribution. Thus, the visibility of the $s$-th subarray for the $k$-th user is modeled as
\begin{equation} 
\left\{
	\begin{aligned}
		& s \in \mathcal{M}_k, {\rm{with\ probability}}\ p,\\
		&s \notin \mathcal{M}_k, {\rm{with\ probability}}\ 1-p,
	\end{aligned}
	\right.
\end{equation}
where $p$ denotes the VR generation probability. From \cite{9733790} and \cite{9777939}, we assume $p=0.35$. Furthermore, the impact of this parameter on the proposed algorithm will be discussed in the simulations.
\vspace{-0.3cm}
\subsection{Algorithm Performance Comparison}
\begin{figure}[t]
	\centering
	\includegraphics[width=0.75\linewidth]{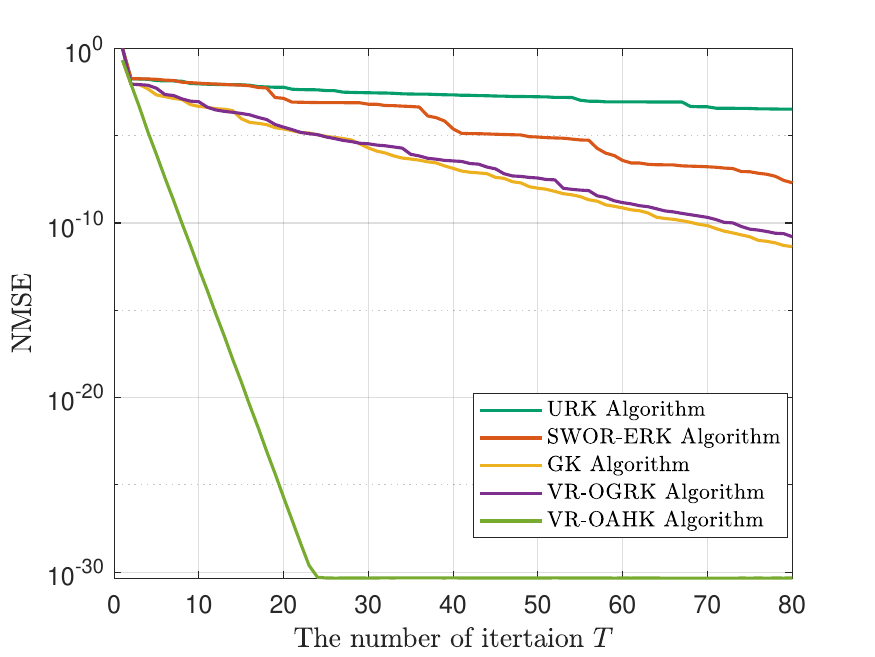}
	\captionsetup{name={Fig.},labelsep=period,singlelinecheck=off}
	\caption{{Convergence behaviors of different algorithms.}}\vspace{-0.55cm}
	\label{fig2}
\end{figure}
\begin{figure}[t]
	\centering
	\includegraphics[width=0.75\linewidth]{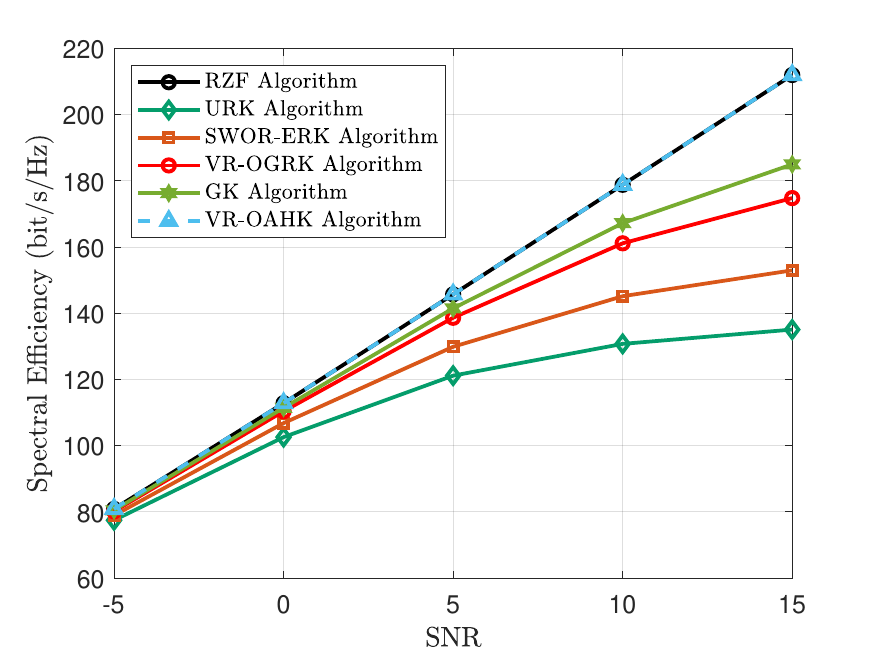}
	\captionsetup{name={Fig.},labelsep=period,singlelinecheck=off}
	\caption{{Spectral efficiency versus $\rm{SNR}$ with $T=15$.}}\vspace{-0.55cm}
	\label{fig3}
\end{figure}
Fig. \ref{fig2} compares the convergence behaviors of different algorithms in terms of NMSE ($ \textrm{NMSE}=\frac{\|\mathbf{F}^{\textrm{RZF}}-{\mathbf{F}}\|_F}{\|\mathbf{F}^{\textrm{RZF}}\|_F}$)\footnote{The NMSE value approaches zero, indicating that the sum rate realized by this scheme is nearly equivalent to that of the RZF scheme.}. As shown in Fig. \ref{fig2}, the proposed VR-OGRK algorithm exhibits a faster convergence rate than both the URK and SWOR-ERK algorithms, while its convergence behavior is similar to that of the GK algorithm. Furthermore, to achieve the NMSE of $10^{-6}$, the proposed VR-OGRK and VR-OAHK algorithms require only $27$ and $5$ iterations, respectively, representing a significant reduction compared to the SWOR-ERK and URK algorithms. For the VR-OGRK algorithm, the improvement of the convergence speed is attributed to its selection strategy, which ensures that the hyperplane with a large residual error is chosen at each iteration. Building on this, the proposed VR-OAHK algorithm further leverages information from multiple hyperplanes rather than a single hyperplane at each iteration, resulting in a faster convergence rate.

In Fig. \ref{fig3}, we investigate the spectral efficiency versus SNR. It is important to note that we assume the iteration number of $T=T^{\rm{URK}}=T^{\rm{ERK}}=T^{\rm{GK}}=T^{\rm{OGRK}}=T^{\rm{OAHK}}=15$. As observed from Fig. \ref{fig3}, the spectral efficiency of the proposed VR-OAHK algorithm is nearly equal to that of the RZF algorithm after a small number of iterations. While the spectral efficiency of the proposed VR-OGRK algorithm is slightly lower than that of the VR-OAHK algorithm, it remains significantly higher than that of the SWOR-ERK and URK algorithms. This is because the two proposed algorithms have faster convergence rate by leveraging the VR features. Furthermore, as SNR increases, traditional Kaczmarz-type algorithms do not yield satisfactory performance. This is because communication scenarios with high SNR require higher accuracy in matrix inversion, and thus the traditional Kaczmarz algorithm requires an extra number of iterations to approximate the matrix inversion. Thus, the slower convergence of the traditional Kaczmarz-type algorithm leads to reduced spectral efficiency at a fixed number of iterations. Moreover, when the number of iterations is unrestricted, all algorithms can achieve similar performance to RZF. However, the slow convergence rate of the traditional Kaczmarz algorithm results in additional computational overhead, which will be discussed later.

\begin{figure}[t]
	\centering
	\includegraphics[width=0.75\linewidth]{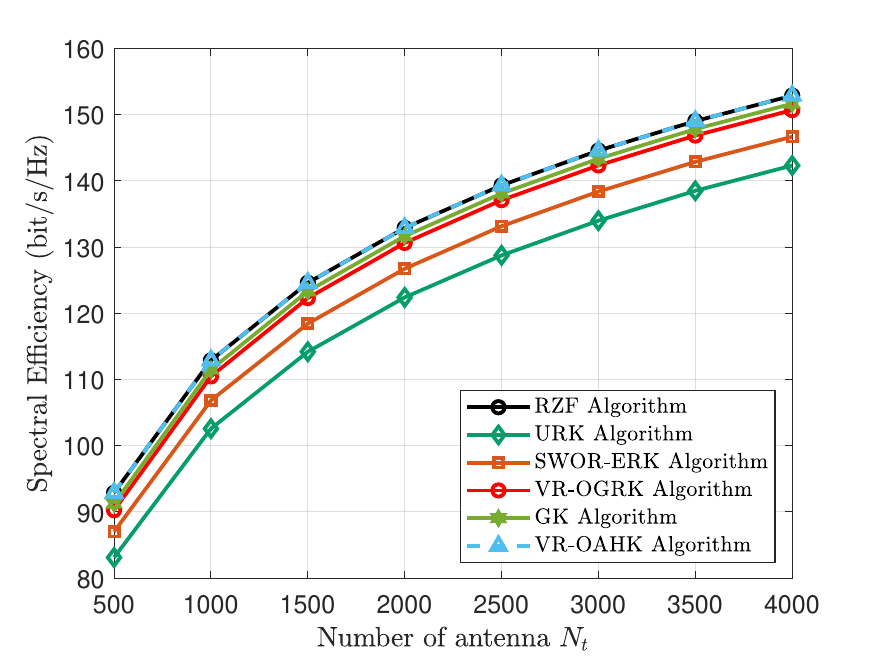}
	\captionsetup{name={Fig.},labelsep=period,singlelinecheck=off}
	\caption{ {Spectral efficiency versus $N_t$ with $T=15$.}}\vspace{-0.5cm}
	\label{fig4}
\end{figure}

\begin{figure}[t]
	\centering
	\includegraphics[width=0.75\linewidth]{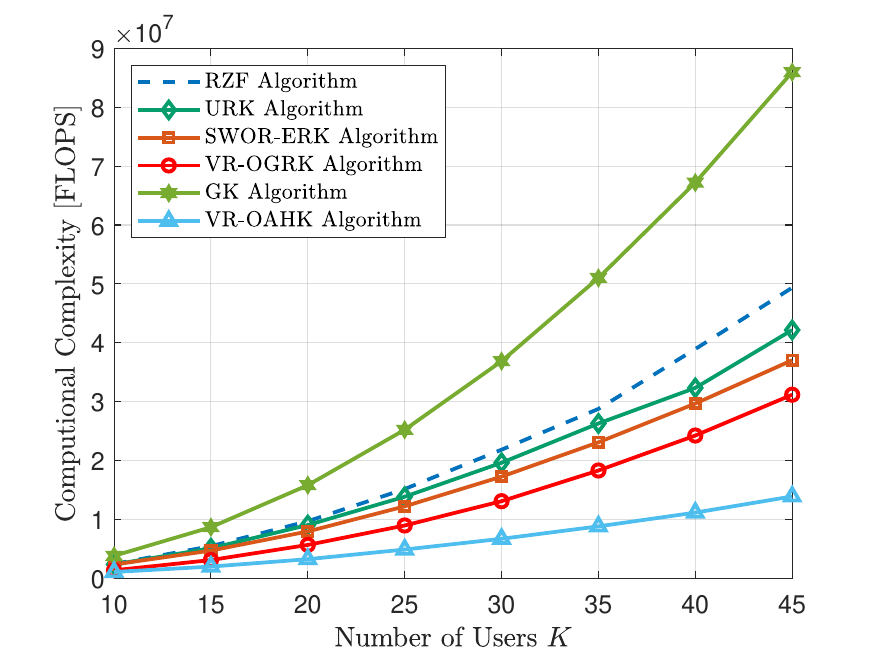}
	\captionsetup{name={Fig.},labelsep=period,singlelinecheck=off}
	\caption{{Computational complexity versus $K$.}}\vspace{-0.55cm}
	\label{fig5}
\end{figure}
\begin{figure}[t]
	\centering
	\includegraphics[width=0.75\linewidth]{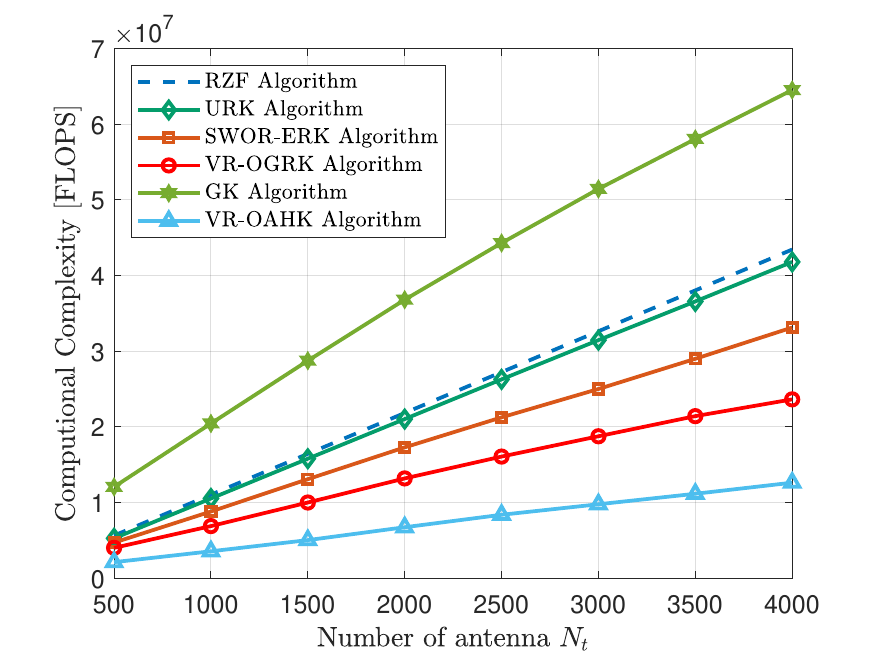}
	\captionsetup{name={Fig.},labelsep=period,singlelinecheck=off}
	\caption{{Computational complexity versus $N_t$.}}\vspace{-0.5cm}
	\label{fig6}
\end{figure}
\begin{figure}[t]
	\centering
	\includegraphics[width=0.75\linewidth]{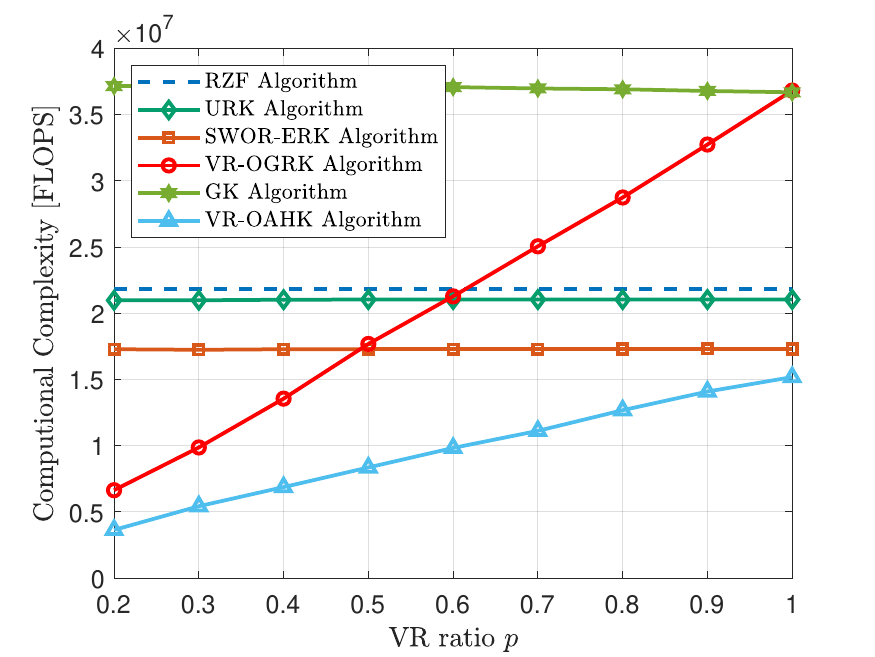}
	\captionsetup{name={Fig.},labelsep=period,singlelinecheck=off}
	\caption{{Computational complexity versus $p$.}}\vspace{-0.55cm}
	\label{fig7}
\end{figure}
Fig. \ref{fig4} illustrates the relationship between the spectral efficiency and the number of antenna $N_t$ with a fixed iteration number of $T=15$ across different algorithms. As illustrated in Fig. \ref{fig4}, the proposed VR-OGRK algorithm and VR-OAHK algorithm can achieve the spectral efficiency comparable to the RZF algorithm. Furthermore, as the number of antennas $N_t$ increases, both proposed algorithms continue to maintain similar performance to the RZF algorithm. This is because both the proposed algorithms converge faster than the SWOR-ERK and URK algorithm. Consequently, with a fixed number of iterations, the two proposed algorithms exhibit a lower NMSE, resulting in higher spectral efficiency. Therefore, even with a larger number of antennas, the two proposed algorithms can achieve the spectral efficiencies close to that of RZF.

In Fig. \ref{fig5}, we compare the computational complexity of different algorithms in terms of $K$. As observed in Fig. \ref{fig5}, the proposed VR-OGRK algorithm and the proposed VR-OAHK algorithm have a lower computational complexity than the traditional Kaczmarz-type algorithms. The complexity of the GK algorithm exceeds that of the RZF scheme, mainly due to the additional residual computation during each iteration. It can be seen that the computational complexity of the proposed VR-OGRK algorithm increases at a slower rate as $K$ rises. This is because the complexity with respect to the number of users in the proposed VR-OGRK algorithm scales as $ \mathcal{O}(8N_tK|{\mathcal{Q}}|)$, while for the URK and SWOR-ERK algorithms, it is on the order of $\mathcal{O}(8N_tK^2)$. Since the growth of $|{\mathcal{Q}}|$ is less than that of $K$, the proposed VR-OGRK algorithm maintains a low computational complexity even with a large total number of users. Additionally, the proposed VR-OAHK algorithm further reduces the computational cost by decreasing the required number of iterations using  information from multiple hyperplanes in each iteration.

Fig. \ref{fig6} examines the computational complexity of different algorithms in relation to $N_t$. It can be observed that the proposed VR-OGRK and VR-OAHK algorithms substantially reduce the computational complexity compared to the other algorithms. Moreover, as $N_t$ increases, the proposed algorithms maintain low computational complexities. This is because the computational complexity of $N_t$ in the VR-OGRK algorithm is $\mathcal{O}(T(8|{{\Gamma}}|(|{\mathcal{X}}|+1)))$, whereas for the URK and SWOR-ERK algorithms, it is $\mathcal{O}(T(16N_t))$. Due to the spatial non-stationarity effect, the increase in $|\Gamma|$ is less than that in $N_t$. Furthermore, by utilizing information from multiple hyperplanes in each iteration, the proposed VR-OAHK algorithm achieves a lower computational complexity than the VR-OGRK algorithm. Thus, as shown in Fig. \ref{fig3}-\ref{fig6}, the two proposed algorithms can achieve satisfactory performance with low complexity in ELAA systems.

In Fig. \ref{fig7}, we explore the relationship between the computational complexity and the VR generation probability $p$. As the value of $p$ increases, the computational complexity of both the proposed VR-OGRK algorithm and the proposed VR-OAHK algorithm also increases. A larger $p$ indicates that the spatial non-stationary effect of the communication system is less pronounced. Therefore, the computational complexity reduction achieved by the proposed VR-OGRK algorithm by utilizing the spatial non-stationary effect will be weakened, and will eventually degenerate to the performance of standard GK algorithm. Fortunately, owing to the fast convergence rate, the proposed VR-OAHK algorithm substantially reduces the computational complexity even at high values of $p$ compared to the SWOR-ERK and URK algorithm. This is because the proposed VR-OAHK algorithm utilizes information from multiple hyperplanes, in contrast to the VR-OGRK algorithm and traditional Kaczmarz-type algorithms, which rely on a single hyperplane. Therefore, the proposed VR-OAHK algorithm possesses promising potentials to enable  future ELAA systems even in case of weak spatial non-stationary effects.
\vspace{-0.3cm}
\section{Conclusion}
In this paper, we proposed a low-complexity iterative precoding design for near-field multiuser systems. Firstly, we reformulated the precoding design problem as a least square problem, which can be addressed by utilizing the iterative Kaczmarz algorithm. Next, to address the limitations of existing algorithms and leverage the properties of spatial non-stationarity, we proposed a VR-OGRK algorithm to enhance convergence rate and significantly reduce computational complexity. We also proposed a novel VR-OAHK algorithm based on graph theory and the idea of aggregation hyperplane to further reduce computational complexity. Furthermore, the two proposed algorithms were shown to exhibit global exponential convergence. Simulation results demonstrated the effectiveness of the two proposed low-complexity algorithms.
\section*{Appendix A}
Consider the following optimization problem
	\begin{align} \label{lemma1_problem}
			\mathbf{v}=\arg \min\limits_{\mathbf{v}} f(\mathbf{v})= \| \mathbf{G}\mathbf{v} -\mathbf{g}\|^2_2,
		\end{align}
	where $\mathbf{G}=[\mathbf{H};\sqrt{\xi} {\mathbf{I}}_K] \in \mathbb{C}^{(N_t+K) \times K}$, $\mathbf{g}=[\mathbf{0}_{N_t};\frac{\mathbf{s}}{\sqrt{\xi}}] \in \mathbb{C}^{(N_t \times K)}$.
	Since $f(\mathbf{v})$ is convex with $\mathbf{v}$, by checking the optimality condition, the optimal solution of $\mathbf{v}$ can be expressed as
	\begin{align}
			&\frac{\partial{f(\mathbf{v})}}{\partial{\mathbf{v}}}=2\mathbf{G}^{\rm{H}}\mathbf{G}\mathbf{v}-2\mathbf{G}^{\rm{H}}\mathbf{g}=\mathbf{0} \nonumber \\
			&\Rightarrow	\mathbf{v}=(\mathbf{H}^{\rm{H}}{\mathbf{H}}+\xi{\mathbf{I}_{K}})^{-1}\mathbf{s}	.
		\end{align}
Hence, the proof of Theorem 1 is completed.

	\vspace{-0.3cm}
	\section*{Appendix B}
	At the $t$-th iteration, we assume that the $i$-th row of $\mathbf{A}$, i.e., $\mathbf{a}^{\rm{H}}_i$ is selected, and then the solution for the next iteration $\mathbf{x}^{t+1}$ is updated as follows
	\begin{align}\label{update2}
		\mathbf{x}^{(t+1)}=\mathbf{x}^{(t)}+\frac{{b}_i-\mathbf{a}^{\rm{H}}_i\mathbf{x}^{(t)}}{\|\mathbf{a}^{\rm{H}}_i\|^2_2}\mathbf{a}_i.
	\end{align}
	
	Next, if the $j$-th row of $\mathbf{A}$, i.e., $\mathbf{a}^{\rm{H}}_j$ which is orthogonal to $\mathbf{a}^{\rm{H}}_i$, i.e., $\mathbf{a}_j^{\rm{H}}\mathbf{a}_i=0$, the residual for $\mathbf{a}^{\rm{H}}_j$ at the $t+1$ iteration is computed as follows
	\begin{align}
		{{r}}^{(t+1)}_j&={{b}_j-\mathbf{a}^{\rm{H}}_j\mathbf{x}^{(t+1)}} \nonumber \\
		&\overset{(b)}{=}{{b}_j-\mathbf{a}^{\rm{H}}_j(\mathbf{x}^{(t)}+\frac{{b}_i-\mathbf{a}^{\rm{H}}_i\mathbf{x}^{(t)}}{\|\mathbf{a}^{\rm{H}}_i\|^2_2}\mathbf{a}_i)} \nonumber \\
		&\overset{(c)}{=}{{b}_j-\mathbf{a}^{\rm{H}}_j\mathbf{x}^{(t)}}=	{{r}}^{(t)}_j,
	\end{align}
	where $(b)$ denotes the update criteria (\ref{update2}), $(c)$ uses the orthogonal property between $\mathbf{a}_i$ and $\mathbf{a}_j$. Clearly, after selecting the $i$-th row of $\mathbf{A}$, only the residuals corresponding to the rows that are not orthogonal to the $i$-th row are affected.
Hence, the proof of Theorem 2 is completed.
	\vspace{-0.25cm}
	\section*{Appendix C}
 We assume that there exists a set of mutually orthogonal hyperplanes in the linear system, denoted as $\mathcal{F}$, i.e., $\mathbf{a}_j^{\rm{H}}\mathbf{a}_i=0, \forall i, j \in \mathcal{F}, i \neq j$.
	At the $t$-th iteration, the aggregation hyperplane is generated by the hyperplanes in $\mathcal{F}$, i.e.,
	\begin{equation} 
		 [{\boldsymbol{\varphi}}^{(t)}]_i=\left\{
		\begin{aligned}
			&\frac{{b}_i-\mathbf{a}^{\rm{H}}_i\mathbf{x}^{(t)}}{\|\mathbf{a}_i^{\rm{H}}\|^2_2}, i\in \mathcal{F},\\
			&0, i \notin \mathcal{F}.
		\end{aligned}
		\right.
	\end{equation}
	
	Thus, the next iteration solution $\mathbf{x}^{(t+1)}$ is calculated as 
	\begin{align}
	\mathbf{x}^{(t+1)}=\mathbf{x}^{(t)}+\frac{(\boldsymbol{\varphi}^{(t)})^{\rm{H}}\mathbf{r}^{(t)}}{\|\mathbf{A}^{\rm{H}}\boldsymbol{\varphi}^{(t)}\|^2_2}\mathbf{A}^{\rm{H}}\boldsymbol{\varphi}^{(t)}.
	\end{align}
	
To demonstrate the relationship
 $\mathbf{A}_{[\mathcal{F},:]}\mathbf{x}^{(t+1)}=\mathbf{b}_{\mathcal{F}}$, we need to prove that $\mathbf{a}^{\rm{H}}_j\mathbf{x}^{(t+1)}=b_j, \forall j \in \mathcal{F}$, as demonstrated below:
	\begin{align}
		\mathbf{a}^{\rm{H}}_j\mathbf{x}^{(t+1)}&=	\mathbf{a}^{\rm{H}}_j\mathbf{x}^{(t)}+	\mathbf{a}^{\rm{H}}_j\frac{(\boldsymbol{\varphi}^{(t)})^{\rm{H}}\mathbf{r}^{(t)}}{\|\mathbf{A}^{\rm{H}}\boldsymbol{\varphi}^{(t)}\|^2_2}\mathbf{A}^{\rm{H}}\boldsymbol{\varphi}^{(t)} \nonumber \\
		&= \mathbf{a}^{\rm{H}}_j\mathbf{x}^{(t)}+\frac{(\boldsymbol{\varphi}^{(t)})^{\rm{H}}\mathbf{r}^{(t)}}{\|\mathbf{A}^{\rm{H}}\boldsymbol{\varphi}^{(t)}\|^2_2}\mathbf{a}^{\rm{H}}_j\mathbf{A}^{\rm{H}}\boldsymbol{\varphi}^{(t)} \nonumber \\
		&= \mathbf{a}^{\rm{H}}_j\mathbf{x}^{(t)}+\frac{\sum_{i\in \mathcal{F}}^{|{\mathcal{F}}|} [{\boldsymbol{\varphi}}^{(t)}]_i{r}^{(t)}_i}{\|\sum_{i\in \mathcal{F}}^{|{\mathcal{F}}|} [{\boldsymbol{\varphi}}^{(t)}]_i\mathbf{a}_i\|^2_2}\mathbf{a}^{\rm{H}}_j \sum_{i\in \mathcal{F}}^{|{\mathcal{F}}|} [{\boldsymbol{\varphi}}^{(t)}]_i\mathbf{a}_i.
	\end{align}
	
By using the property of orthogonality, we can establish the following relationship
	\begin{align}\label{eq1}
	\mathbf{a}^{\rm{H}}_j\mathbf{x}^{(t+1)}&=\mathbf{a}^{\rm{H}}_j\mathbf{x}^{(t)}+\frac{\sum_{i\in \mathcal{F}}^{|{\mathcal{F}}|} [{\boldsymbol{\varphi}}^{(t)}]_i{r}^{(t)}_i}{\sum_{i\in \mathcal{F}}^{|{\mathcal{F}}|}\| [{\boldsymbol{\varphi}}^{(t)}]_i\mathbf{a}_i\|^2_2} [{\boldsymbol{\varphi}}^{(t)}]_j\|\mathbf{a}_j\|_2^2.
	\end{align}
	
Next, by substituting (\ref{coefficient}) into (\ref{eq1}), we can arrive at the following result
	\begin{align} 
	\mathbf{a}^{\rm{H}}_j\mathbf{x}^{(t+1)}	&=\mathbf{a}^{\rm{H}}_j\mathbf{x}^{(t)}+\frac{\sum_{i\in \mathcal{F}}^{|{\mathcal{F}}|} [{\boldsymbol{\varphi}}^{(t)}]_i{r}^{(t)}_i}{\sum_{i\in \mathcal{F}}^{|{\mathcal{F}}|} [{\boldsymbol{\varphi}}^{(t)}]_i{r}^{(t)}_i} [{\boldsymbol{\varphi}}^{(t)}]_j\|\mathbf{a}_j\|_2^2 \nonumber \\
	&=\mathbf{a}^{\rm{H}}_j\mathbf{x}^{(t)}+{r}^{(t)}_j \nonumber \\
	&=\mathbf{a}^{\rm{H}}_j\mathbf{x}^{(t)}+(b_j-\mathbf{a}^{\rm{H}}_j\mathbf{x}^{(t)})=b_j.
	\end{align}
	
	Hence, the proof of Theorem 4 is completed.
	 \vspace{-0.3cm}
	\section*{Appendix D}
We assume that there exists a set of mutually orthogonal hyperplanes in the linear system, denoted as $\mathcal{F}$, i.e., $\mathbf{a}_j^{\rm{H}}\mathbf{a}_i=0, \forall i, j \in \mathcal{F}, i \neq j$.
	At the $t$-th iteration, the aggregation hyperplane is generated by the hyperplanes in $\mathcal{F}$, i.e.,
	\begin{equation} 
		 [{\boldsymbol{\varphi}}^{(t)}]_i=\left\{
		\begin{aligned}
			&\frac{{b}_i-\mathbf{a}^{\rm{H}}_i\mathbf{x}^{(t)}}{\|\mathbf{a}_i^{\rm{H}}\|^2_2}, i\in \mathcal{F},\\
			&0, i \notin \mathcal{F}.
		\end{aligned}
		\right.
	\end{equation}
	
	Thus, the next iteration solution $\mathbf{x}^{(t+1)}$ is calculated as 
	\begin{align}
		\mathbf{x}^{(t+1)}&=\mathbf{x}^{(t)}+\frac{(\boldsymbol{\varphi}^{(t)})^{\rm{H}}\mathbf{r}^{(t)}}{\|\mathbf{A}^{\rm{H}}\boldsymbol{\varphi}^{(t)}\|^2_2}\mathbf{A}^{\rm{H}}\boldsymbol{\varphi}^{(t)} \nonumber \\
		&=\mathbf{x}^{(t)}+\frac{\sum_{i\in \mathcal{F}}^{|{\mathcal{F}}|} [{\boldsymbol{\varphi}}^{(t)}]_i{r}^{(t)}_i}{\|\sum_{i\in \mathcal{F}}^{|{\mathcal{F}}|} [{\boldsymbol{\varphi}}^{(t)}]_i\mathbf{a}_i\|^2_2} \sum_{i\in \mathcal{F}}^{|{\mathcal{F}}|} [{\boldsymbol{\varphi}}^{(t)}]_i\mathbf{a}_i \nonumber \\
		&=\mathbf{x}^{(t)}+\frac{\sum_{i\in \mathcal{F}}^{|{\mathcal{F}}|} [{\boldsymbol{\varphi}}^{(t)}]_i{r}^{(t)}_i}{\sum_{i\in \mathcal{F}}^{|{\mathcal{F}}|} [{\boldsymbol{\varphi}}^{(t)}]_i{r}^{(t)}_i} \sum_{i\in \mathcal{F}}^{|{\mathcal{F}}|} [{\boldsymbol{\varphi}}^{(t)}]_i\mathbf{a}_i \nonumber \\
		&=\mathbf{x}^{(t)}+\sum_{i\in \mathcal{F}}^{|{\mathcal{F}}|} [{\boldsymbol{\varphi}}^{(t)}]_i\mathbf{a}_i.
	\end{align}
	
	Hence, the proof of Proposition 1 is completed.
	\vspace{-0.3cm}
	\bibliographystyle{IEEEtran}
	\bibliography{ref}

\begin{thebibliography}{10}
\providecommand{\url}[1]{#1}
\csname url@samestyle\endcsname
\providecommand{\newblock}{\relax}
\providecommand{\bibinfo}[2]{#2}
\providecommand{\BIBentrySTDinterwordspacing}{\spaceskip=0pt\relax}
\providecommand{\BIBentryALTinterwordstretchfactor}{4}
\providecommand{\BIBentryALTinterwordspacing}{\spaceskip=\fontdimen2\font plus
\BIBentryALTinterwordstretchfactor\fontdimen3\font minus
  \fontdimen4\font\relax}
\providecommand{\BIBforeignlanguage}[2]{{%
\expandafter\ifx\csname l@#1\endcsname\relax
\typeout{** WARNING: IEEEtran.bst: No hyphenation pattern has been}%
\typeout{** loaded for the language `#1'. Using the pattern for}%
\typeout{** the default language instead.}%
\else
\language=\csname l@#1\endcsname
\fi
#2}}
\providecommand{\BIBdecl}{\relax}
\BIBdecl

\bibitem{you2021towards}
X.~You \emph{et~al.}, ``Towards {6G} wireless communication networks: Vision,
  enabling technologies, and new paradigm shifts,'' \emph{Sci. China Inf.
  Sci.}, vol.~64, pp. 1--74, Nov. 2021.

\bibitem{8766143}
Z.~Zhang \emph{et~al.}, ``{6G} wireless networks: Vision, requirements,
  architecture, and key technologies,'' \emph{IEEE Veh. Technol. Mag.},
  vol.~14, no.~3, pp. 28--41, Sep. 2019.

\bibitem{10054381}
C.-X. Wang \emph{et~al.}, ``On the road to {6G}: Visions, requirements, key
  technologies, and testbeds,'' \emph{IEEE Commun. Surveys Tuts.}, vol.~25,
  no.~2, pp. 905--974, 2nd Quart 2023.

\bibitem{9237116}
------, ``{6G} wireless channel measurements and models: Trends and
  challenges,'' \emph{IEEE Veh. Technol. Mag.}, vol.~15, no.~4, pp. 22--32,
  Dec. 2020.

\bibitem{6732923}
S.~Rangan, T.~S. Rappaport, and E.~Erkip, ``Millimeter-wave cellular wireless
  networks: Potentials and challenges,'' \emph{Proc. IEEE}, vol. 102, no.~3,
  pp. 366--385, Mar. 2014.

\bibitem{7400949}
R.~W. Heath \emph{et~al.}, ``An overview of signal processing techniques for
  millimeter wave {MIMO} systems,'' \emph{IEEE J. Sel. Topics Signal Process.},
  vol.~10, no.~3, pp. 436--453, Apr. 2016.

\bibitem{6005345}
H.-J. Song and T.~Nagatsuma, ``Present and future of {Terahertz}
  communications,'' \emph{IEEE Trans. THz Sci. Technol.}, vol.~1, no.~1, pp.
  256--263, Sep. 2011.

\bibitem{8387211}
I.~F. Akyildiz, C.~Han, and S.~Nie, ``Combating the distance problem in the
  millimeter wave and {Terahertz} frequency bands,'' \emph{IEEE Commun. Mag.},
  vol.~56, no.~6, pp. 102--108, Jun. 2018.

\bibitem{9475160}
C.~Pan \emph{et~al.}, ``Reconfigurable intelligent surfaces for {6G} systems:
  Principles, applications, and research directions,'' \emph{IEEE Commun.
  Mag.}, vol.~59, no.~6, pp. 14--20, Jun. 2021.

\bibitem{9424177}
Y.~Liu, X.~Liu, X.~Mu, T.~Hou, J.~Xu, M.~Di~Renzo, and N.~Al-Dhahir,
  ``Reconfigurable intelligent surfaces: Principles and opportunities,''
  \emph{IEEE Commun. Surveys Tuts.}, vol.~23, no.~3, pp. 1546--1577, 3rd Quart.
  2021.

\bibitem{9326394}
Q.~Wu, S.~Zhang, B.~Zheng, C.~You, and R.~Zhang, ``Intelligent reflecting
  surface-aided wireless communications: A tutorial,'' \emph{IEEE Trans.
  Commun.}, vol.~69, no.~5, pp. 3313--3351, May. 2021.

\bibitem{10098681}
Z.~Wang \emph{et~al.}, ``Extremely large-scale {MIMO}: Fundamentals,
  challenges, solutions, and future directions,'' \emph{IEEE Wireless Commun.},
  vol.~31, no.~3, pp. 117--124, Jun. 2024.

\bibitem{BJORNSON20193}
E.~Björnson and oters, ``Massive {MIMO} is a reality—what is next?: Five
  promising research directions for antenna arrays,'' \emph{Digit. Signal
  Process.}, vol.~94, pp. 3--20, Nov. 2019.

\bibitem{9170653}
F.~Tariq \emph{et~al.}, ``A speculative study on {6G},'' \emph{IEEE Wireless
  Commun.}, vol.~27, no.~4, pp. 118--125, Aug. 2020.

\bibitem{8869705}
W.~Saad, M.~Bennis, and M.~Chen, ``A vision of {6G} wireless systems:
  Applications, trends, technologies, and open research problems,'' \emph{IEEE
  Netw.}, vol.~34, no.~3, pp. 134--142, May. 2020.

\bibitem{9170651}
E.~D. Carvalho, A.~Ali, A.~Amiri, M.~Angjelichinoski, and R.~W. Heath,
  ``Non-stationarities in extra-large-scale massive {MIMO},'' \emph{IEEE
  Wireless Commun.}, vol.~27, no.~4, pp. 74--80, Aug. 2020.

\bibitem{6415388}
J.~Hoydis, S.~ten Brink, and M.~Debbah, ``Massive {MIMO} in the {UL}/{DL} of
  cellular networks: How many antennas do we need?'' \emph{IEEE J. Sel. Areas
  Commun.}, vol.~31, no.~2, pp. 160--171, Feb. 2013.

\bibitem{8169014}
N.~Fatema \emph{et~al.}, ``Massive {MIMO} linear precoding: A survey,''
  \emph{IEEE syst. J.}, vol.~12, no.~4, pp. 3920--3931, Dec. 2018.

\bibitem{7248580}
D.~Zhu, B.~Li, and P.~Liang, ``On the matrix inversion approximation based on
  neumann series in massive {MIMO} systems,'' in \emph{Proc IEEE ICC}, Jun.
  2015, pp. 1763--1769.

\bibitem{8064675}
B.~Nagy, M.~Elsabrouty, and S.~Elramly, ``Fast converging weighted neumann
  series precoding for massive {MIMO} systems,'' \emph{IEEE Wireless Commun.
  Lett.}, vol.~7, no.~2, pp. 154--157, Apr. 2018.

\bibitem{8425997}
M.~N. Boroujerdi \emph{et~al.}, ``Low-complexity statistically robust
  precoder/detector computation for massive {MIMO} systems,'' \emph{IEEE Trans.
  Wireless Commun.}, vol.~17, no.~10, pp. 6516--6530, Oct. 2018.

\bibitem{8417575}
J.~Minango and C.~de~Almeida, ``A low-complexity linear precoding algorithm
  based on {Jacobi} method for massive {MIMO} systems,'' in \emph{Proc. IEEE
  87th Veh. Technol. Conf. (VTC Spring)}, Jun. 2018, pp. 1--5.

\bibitem{7146060}
Z.~Lu, J.~Ning, Y.~Zhang, T.~Xie, and W.~Shen, ``Richardson method based linear
  precoding with low complexity for massive {MIMO} systems,'' in \emph{Proc.
  IEEE 81st Veh. Technol. Conf. (VTC Spring)}, May. 2015, pp. 1--4.

\bibitem{10220155}
Z.~Wang \emph{et~al.}, ``Rapidly converging low-complexity iterative transmit
  precoders for massive {MIMO} downlink,'' \emph{IEEE Trans. Commun.}, vol.~71,
  no.~12, pp. 7228--7243, Dec. 2023.

\bibitem{10.1007/11830924_45}
T.~Strohmer and R.~Vershynin, ``A randomized kaczmarz algorithm with
  exponential convergence,'' \emph{J. Fourier Anal. Appl.}, vol.~15, no.~2, pp.
  262--278, Apr. 2009.

\bibitem{9437708}
V.~Croisfelt \emph{et~al.}, ``Accelerated randomized methods for receiver
  design in extra-large scale {MIMO} arrays,'' \emph{IEEE Trans. Veh.
  Technol.}, vol.~70, no.~7, pp. 6788--6799, Jul. 2021.

\bibitem{strohmer2009comments}
T.~Strohmer and R.~Vershynin, ``Comments on the randomized kaczmarz method,''
  \emph{J. Fourier Anal. Appl.}, vol.~15, no.~4, pp. 437--440, 2009.

\bibitem{censor2009note}
Y.~Censor, G.~T. Herman, and M.~Jiang, ``A note on the behavior of the
  randomized kaczmarz algorithm of strohmer and vershynin,'' \emph{J. Fourier
  Anal. Appl.}, vol.~15, pp. 431--436, 2009.

\bibitem{10220205}
Y.~Liu, Z.~Wang, J.~Xu, C.~Ouyang, X.~Mu, and R.~Schober, ``Near-field
  communications: A tutorial review,'' \emph{IEEE Open J. Commun. Soc.},
  vol.~4, pp. 1999--2049, Aug. 2023.

\bibitem{lu2023tutorial}
H.~Lu \emph{et~al.}, ``A tutorial on near-field {XL-MIMO} communications
  towards {6G},'' \emph{IEEE Commun. Surveys Tuts.}, pp. 1--1, 2024.

\bibitem{10225614}
C.-X. Wang \emph{et~al.}, ``A complete study of space-time-frequency
  statistical properties of the {6G} pervasive channel model,'' \emph{IEEE
  Trans. Commun.}, vol.~71, no.~12, pp. 7273--7287, Dec 2023.

\bibitem{10500425}
K.~Zhi \emph{et~al.}, ``Performance analysis and low-complexity design for
  {XL-MIMO} with near-field spatial non-stationarities,'' \emph{IEEE J. Sel.
  Areas Commun.}, vol.~42, no.~6, pp. 1656--1672, Jun. 2024.

\bibitem{9145378}
V.~C. Rodrigues, A.~Amiri, T.~Abrão, E.~de~Carvalho, and P.~Popovski,
  ``Low-complexity distributed {XL-MIMO} for multiuser detection,'' in
  \emph{Proc IEEE ICC Workshops}, Jun. 2020, pp. 1--6.

\bibitem{10279327}
B.~Xu, Z.~Wang, H.~Xiao, J.~Zhang, B.~Ai, and D.~W. Kwan~Ng, ``Low-complexity
  precoding for extremely large-scale {MIMO} over non-stationary channels,'' in
  \emph{Proc IEEE ICC}, May. 2023, pp. 6516--6522.

\bibitem{yuan2022spatial}
Z.~Yuan \emph{et~al.}, ``Spatial non-stationary near-field channel modeling and
  validation for massive {MIMO} systems,'' \emph{IEEE Trans. Antennas Propag.},
  vol.~71, no.~1, pp. 921--933, Jan. 2023.

\bibitem{9786750}
C.-X. Wang \emph{et~al.}, ``Pervasive wireless channel modeling theory and
  applications to {6G} gbsms for all frequency bands and all scenarios,''
  \emph{IEEE Trans. Veh. Technol.y}, vol.~71, no.~9, pp. 9159--9173, Sep. 2022.

\bibitem{Randomized}
A.~Zouzias and N.~M. Freris, ``Randomized extended kaczmarz for solving least
  squares,'' \emph{SIAM J. Matrix Anal. Appl.}, vol.~34, no.~2, pp. 773--793,
  2013.

\bibitem{GRIEBEL20121596}
M.~Griebel and P.~Oswald, ``Greedy and randomized versions of the
  multiplicative {Schwarz} method,'' \emph{Linear Algebra Appl.}, vol. 437,
  no.~7, pp. 1596--1610, Oct. 2012.

\bibitem{doi:10.1137/17M1137747}
Z.~Bai and W.~Wu, ``On greedy randomized kaczmarz method for solving large
  sparse linear systems,'' \emph{SIAM J. Sci. Comput.}, vol.~40, no.~1, pp.
  A592--A606, 2018.

\bibitem{glover1968surrogate}
F.~Glover, ``Surrogate constraints,'' \emph{Oper. Res.}, vol.~16, no.~4, pp.
  741--749, 1968.

\bibitem{glover2003tutorial}
------, ``Tutorial on surrogate constraint approaches for optimization in
  graphs,'' \emph{J. Heuristics}, vol.~9, pp. 175--227, 2003.

\bibitem{10.1145/3035918.3035939}
L.~Chang, W.~Li, and W.~Zhang, ``Computing a near-maximum independent set in
  linear time by reducing-peeling,'' in \emph{Proc. ACM International Conf.
  Management of Data}, 2017, p. 1181–1196.

\bibitem{golub2013matrix}
G.~H. Golub and C.~F. Van~Loan, \emph{Matrix computations}.\hskip 1em plus
  0.5em minus 0.4em\relax The John Hopkins University Press, 2013.

\bibitem{orthogonality}
J.~Nutini \emph{et~al.}, ``Convergence rates for greedy kaczmarz algorithms,
  and faster randomized kaczmarz rules using the orthogonality graph,''
  \emph{Proc. UAI,}, pp. 1--14, 2016.

\bibitem{9733790}
H.~Iimori \emph{et~al.}, ``Joint activity and channel estimation for
  extra-large {MIMO} systems,'' \emph{IEEE Trans. Wireless Commun.}, vol.~21,
  no.~9, pp. 7253--7270, Sep. 2022.

\bibitem{9777939}
Y.~Han \emph{et~al.}, ``Localization and channel reconstruction for extra large
  {RIS}-assisted massive {MIMO} systems,'' \emph{IEEE J. Sel. Topics Signal
  Process.}, vol.~16, no.~5, pp. 1011--1025, Aug. 2022.

\end{thebibliography}

\end{document}